\documentclass[12pt, draftclsnofoot, onecolumn]{IEEEtran}

\usepackage[cmex10]{amsmath}
\usepackage{amssymb}
\usepackage{cite}
\usepackage{amsmath}
\usepackage{amssymb}
\usepackage{wasysym}
\usepackage{booktabs}
\usepackage{paralist}
\usepackage{balance}
\usepackage{subfigure}
\usepackage{multicol}
\usepackage{multirow}
\usepackage{graphicx}
\usepackage{epsfig}
\usepackage{epstopdf}
\usepackage{algorithmic}
\usepackage[ruled,linesnumbered]{algorithm2e}
\usepackage{bm}
\usepackage{amsmath, amssymb,amsthm}
\usepackage{upgreek}
\usepackage{makecell}
\usepackage{caption}
\usepackage{subfigure}
\usepackage{color}
\usepackage{stfloats}
\theoremstyle{definition} 
\theoremstyle{definition} 
\theoremstyle{plain} 
\theoremstyle{plain} \newtheorem{proposition}{Proposition}

\DeclareMathOperator*{\argmin}{argmin}

\begin{document}

\title{Delay and Power Tradeoff with Consideration of Caching Capabilities in Dense Wireless Networks}
%
\author{Hao~Wu,
        Hancheng~Lu,~\IEEEmembership{Member,~IEEE} 
\IEEEcompsocitemizethanks{\IEEEcompsocthanksitem Hao Wu (Email: hwu2014@mail.ustc.edu.cn) and Hancheng Lu (Email: hclu@ustc.edu.cn) are with the Information Network Lab of EEIS Department, University of Science and Technology of China, Hefei 230027, China .

}
}


\maketitle

\begin{abstract}
Enabling caching capabilities in dense small cell networks (DSCNs) has a direct impact on file delivery delay and power consumption. Most existing work studied these two performance metrics separately in cache-enabled DSCNs. However, file delivery delay and power consumption are coupled with each other and cannot be minimized simultaneously. In this paper, we investigate the optimal tradeoff between these two performance metrics. Firstly, we formulate the joint file delivery delay and power consumption optimization (JDPO) problem where power control, user association and file placement are jointly considered. Then we convert it to a form that can be handled by Generalized Benders Decomposition (GBD). With GBD, we decompose the converted JDPO problem into two smaller problems, i.e., primal problem related to power control and master problem related to user association and file placement. An iterative algorithm is proposed and proved to be $\epsilon$-optimal, in which the primal problem and master problem are solved iteratively to approach the optimal solution. To further reduce the complexity of the master problem, an accelerated algorithm based on semi-definite relaxation is proposed. Finally, the simulation results demonstrate that the proposed algorithm can approach the optimal tradeoff between file delivery delay and power consumption.
\end{abstract}

\begin{IEEEkeywords}
Dense small cell networks (DSCNs), caching, file delivery delay, power consumption, generalized benders decomposition (GBD)
\end{IEEEkeywords}

%
\IEEEpeerreviewmaketitle

\section{Introduction}
To cope with the rapid growth of mobile traffic, network densification with small cell deployment has been proposed as a promising technique, which leverages spatial spectrum reuse to increase network capacity \cite{Udn1}\cite{DSCNReuse}. In dense small cell networks (DSCNs), a large amount of small base stations (SBSs, i.e., Pico BS, Femto BS, et al.) connect to the mobile core network through backhaul. As reported in \cite{CISCOVNI}, video traffic has contributed the major portion of mobile traffic. Such bandwidth-intensive traffic makes backhaul links prone to be bottleneck, leading to a backhaul problem. Existing studies have shown that some files, especially video files, have high popularity and are requested by many users \cite{statistic}\cite{statistic1}. By caching these highly popular files at SBSs, backhaul traffic can be greatly reduced and hence the backhaul problem can be alleviated \cite{CCP}\cite{D-F}.

Enabling caching capabilities at SBSs has a direct impact on the performance of DSCNs, especially in terms of file delivery delay and power consumption\cite{TWC17,TMC18,TOC18,CachePower1,CachePower2}. From the perspective of users, lower file delivery delay means better user experience. If the file requested by a user is cached by its associated SBS, the user can get the file directly from this SBS instead of the remote file server. In this case, the file delivery delay is significantly reduced, which results in improvement in user experience \cite{TWC17,TMC18,TOC18}. From the perspective of mobile network operators, caching files in SBSs incurs additional power consumption. Previous studies have shown that power consumption for caching cannot be ignored at SBSs where power is constrained \cite{CachePower1,CachePower2}. Therefore, in cache-enabled DSCNs, both file delivery delay for users and power consumption at SBSs should be reconsidered. Correspondingly, many research attempts have been made to optimize these two performance metrics in consideration of caching capabilities at SBSs.

To minimize file delivery delay, file placement strategies have been carefully studied in cache-enabled DSCNs. The authors in \cite{TWC17} develop both centralized and distributed transmission aware file placement strategies to minimize delay. In \cite{TMC18}, a new caching  architecture is designed for the cooperative transmissions scenario where the file placement problem is analyzed. Due to the limited SBS cache capacity, a joint file placement and bandwidth allocation schemes with the aim of minimizing the file delivery delay is proposed in \cite{TOC18}. In cache-enabled DSCNs,  file placement policy is often combined with user association to minimize file delivery delay. The problem of file placement and tier-based user association to minimize the file delivery delay is analyzed in \cite{CachingDelay}. To further reduce the complexity of caching placement and user association problem, a distributed algorithm with a low complexity is developed in \cite{D-F}.

There also exist some studies on power consumption at SBSs when caches are involved. In \cite{CachePower1}, by analyzing the relation between  energy efficiency and  cache size, the authors provide the condition when   energy efficiency can benefit from caching. In \cite{GC18}, authors think that caching power and transmit power are both important parts of  total power budget in fiber-wireless access networks. To maximize the downlink throughput under the limited power budget, authors jointly consider power allocation and caching strategies. Caching files at SBSs will incur caching power consumption  while backhaul   power is consumed when files are not cached. Then literature \cite{CachePower2} studies such power consumption tradeoff between  caching and backhaul transmission.

In the aforementioned work, file delivery delay and power consumption,  which are both important performance metrics in cache-enabled DSCNs, have been studied separately. So far very little work has been done to jointly consider these two performance metrics. As file delivery delay and power consumption are coupled with each other, they cannot be minimized simultaneously. Intuitively, to achieve minimum file delivery delay, as much power as possible should be allocated for caching and transmission. It means maximum power should be consumed at SBSs. Hence, there is a tradeoff between these two performance metrics. To achieve the optimal tradeoff, joint optimization of file delivery delay and power consumption should be performed.

However, in cache-enabled DSCNs, the joint file delivery delay and power consumption optimization (JDPO) problem is non-trivial. In traditional DSCNs, the JDPO problem can be solved by jointly power control and user association, whose difficulty largely stems from the coupling relationship caused by inte-cell interference \cite{JDPO1,JDPO2,JDPO3}. With caching capabilities at SBSs, file placement will be an additional flexible variable to the JDPO problem. In this case, file placement should be jointly performed with power control and user association to solve the JDPO problem. Depending on file caching status at SBSs, power consumption for caching should be considered in total power consumption as well as backhaul delay should be considered in file delivery delay. All these make JDPO in cache-enabled DSCNs much more complex than that in traditional DSCNs.

In this paper, we investigate the JDPO problem in cache-enabled DSCNs. To the best of our knowledge, the most similar work to ours is described in \cite{R1}, where the tradeoff between energy consumption and file delivery delay is studied with given file caching status at SBSs. In \cite{R1}, file placement is not jointly performed with power control and user association. Furthermore, power consumption for caching at SBSs is not considered. Based on jointly power control, user association and file placement, we derive the expressions for file delivery delay and power consumption, respectively. Then, we formulate the JDPO problem. To solve the problem, two questions should be answered. The first question is, for each SBS, which power level should be employed for transmitting a requested file. It is related to power control. The second question is, for each user, where to access its requested file, i.e., through which SBS? and then from cache or backhaul of the SBS? It is related to   user association and file placement strategies. Based on these two questions, the JDPO problem can be decomposed into two subproblems and thus its complexity can be reduced.

The main contributions of our work are summarized as follows.
\par
1) We formulate the JDPO problem as a mixed-integer programming (MIP) problem, then decompose it into two subproblems, i.e., transmit power allocation (TPA) problem and file delivery path (FDP) problem. The TPA problem is related to power control at SBSs, while the FDP problem is related to user association and file placement. By relaxing the non-convex TPA problem to a convex one with the tight approximation, we convert the JDPO problem into a form that can be handled by Generalized Benders Decomposition (GBD).
\par
2) With GBD, we decompose the converted JDPO problem into two smaller problems, i.e., primal problem and master problem. The primal problem corresponds to the convex relaxation of the TPC problem, which provides an upper bound of the converted JPDO problem. The master problem corresponds to the FDP problem, which provides a lower bound of the converted JDPO problem.
\par
3) Based on the GBD approach, we propose an iterative algorithm to solve the converted JDPO problem. In each iteration, an upper bound and a lower bound are derived by solving the primal problem and the master problem, respectively. We prove that the proposed iterative algorithm can be converged to an $\epsilon$-optimal solution. To further reduce the complexity of the master problem, we propose an accelerated algorithm based on the semi-definite relaxation (SDR) technique.

Simulations are performed to validate our work. The results show the convergency and optimality of the proposed algorithm. Based on the simulation results, we can conclude that, by jointly power control, user association and file placement, the proposed algorithm can approach the optimal tradeoff between file delivery delay and power consumption.

The rest of this paper is organized as follows. The system model is described in Section II. In Section III, the JDPO problem is formulated, where power control, user association and file placement are jointly considered. Furthermore, the JDPO problem is converted to a form that can be handled by GBD. The converted JDPO problem is decomposed into the primal problem and master problem by GBD in Section IV. In Section V, an iterative algorithm is proposed to approach the optimal solution based on GBD. To reduce the complexity of the master problem, an accelerated algorithm based on SDR is proposed. The simulation results are presented and analyzed in Section VI. Finally, the paper is concluded in Section VII.

%
%
%
%

%


\section{system model}
\begin{figure}[t]
\setlength{\abovecaptionskip}{0.cm}

\setlength{\belowcaptionskip}{-0.cm}
  \centering
  \includegraphics[width=0.63\textwidth]{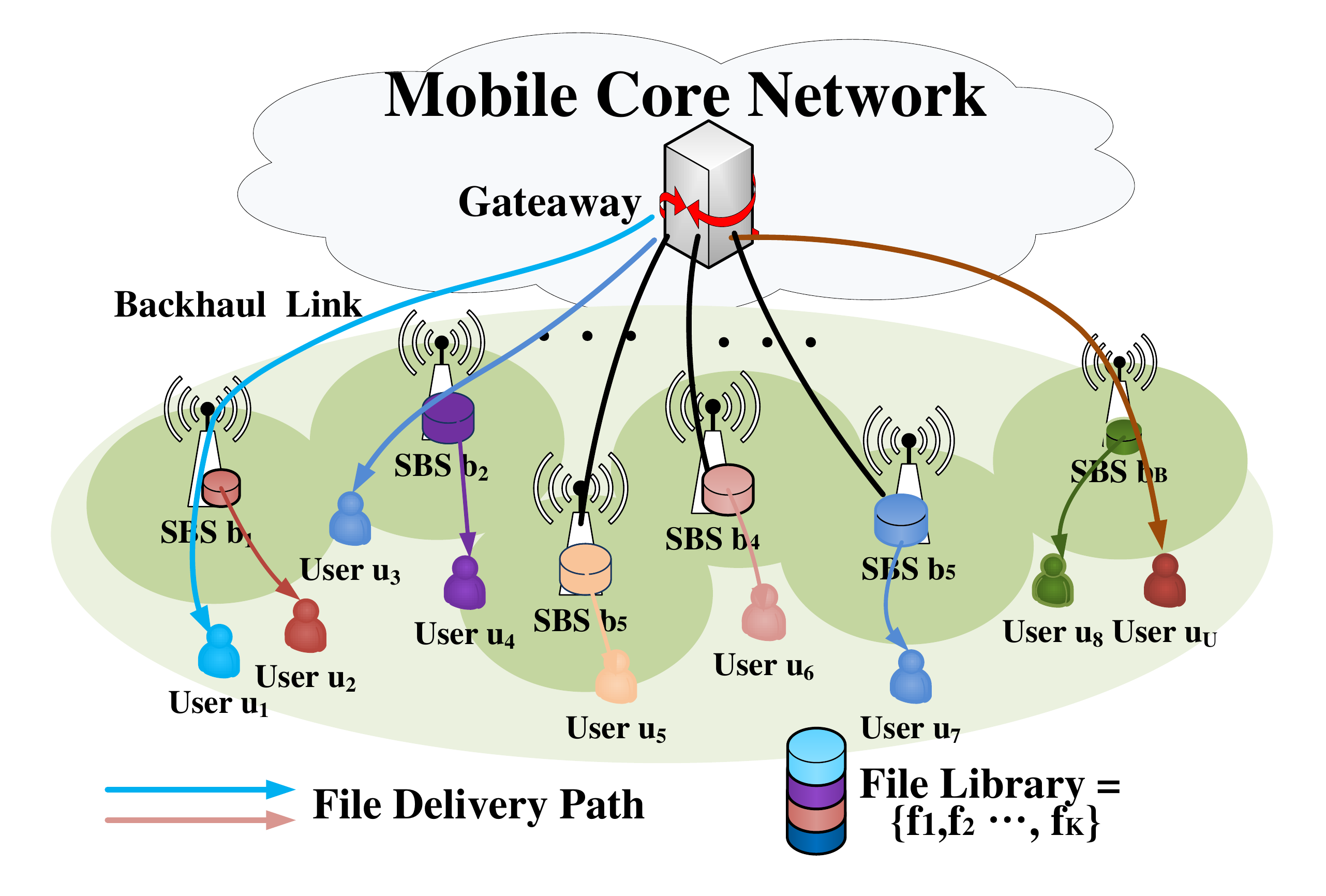}
  \captionsetup{font={small}}
  \caption{Cache-enabled DSCNs architecture.}
  \label{example}
\end{figure}
A downlink DSCN is considered, as shown in Fig. 1 \cite{UDNFig}. In the coverage area of DSCN, there are $B$ small BSs (SBSs ,i.e., femto BSs or pico BSs) indexed by a set $\mathcal{B}$ =\{1, 2, ..., $B$\}.  All SBSs are cache-enabled and the cache capacity of SBS $b_j$ is denoted by $M_j$ (bits) ($j\in \mathcal{B}$). Each SBS is connected to the mobile core network with a capacity-limited backhaul link and the backhaul bandwidth of SBS $b_j$ is $C_j$. $U$ users are randomly deployed. Let $\mathcal{U}$ denote the user index set and $\mathcal{U}$ = \{1, 2, ..., $U$\}. The requested files are indexed by a set $\mathcal{F}$ = \{1, 2, ..., $F$\}, which are stored as a file library at the remote file server. For file $f_k$ ($k\in \mathcal{F}$), its size is denoted by $s_k$(bits). Considering quality of service (QoS) of users, the delivery rate requirement on $f_k$ is denoted by $r_k$(bps). Some major notations are summarized in Table \ref{notation}.

In the cache-enabled DSCN, power consumption and file delivery delay are analyzed and derived as follows.

\renewcommand\arraystretch{0.8}
\begin{table}
  \centering
  \caption{NOTATIONS}
  \label{summary of the notation used in this paper}
  \begin{tabular}{|c|c|}
    \hline
    \multicolumn{2}{|c|}{\textbf{Parameters}}\\
    \hline\textbf{Symbol} & \textbf{Description}\\
    \hline$b_j$, $u_i$, $f_k$ &  SBS, user and file indexed by $j$, $i$, $k$ respectively \\\hline
    $q_{ik}$ &  $u_i$'s file preference for $f_k$\\\hline
    $p_{j}^{ca}$ &  Caching power consumption at $b_j$ \\\hline
    $p_{j}^{cc}$ &  Circuits power consumption at $b_j$ \\\hline
    $p_{j}^{bh}$ &  Backhaul power consumption at $b_j$ \\\hline
    $\sigma_{j}^{ca}$ & Power coefficient of caching hardware in watt/bit\\\hline
    $\sigma_{j}^{bh}$ & Power coefficient of backhaul link in watt/bps\\\hline
    $M_j$ &  Cache capacity of $b_j$ \\\hline
    $C_j$ &  Backhaul bandwidth of $b_j$ \\
    \hline
    \multirow{2}{*}{$\tau_{ijk}^{1},\tau_{jk}^{2}$}  &Wireless transmission delay and Backhaul delay \\
                                                      & of  $f_k$ when  $u_i$ is associated with  $b_j$ \\
    \hline
     $d_{ij}^{k}$ & \makecell{File $f_{k}$ delivery delay for files through from $b_j$ to $u_i$} \\
      \hline
        $p_{ij}^{tr}$ &  Transmit power consumption from $b_j$ to $u_i$\\\hline
    $x_{ij}\in \{0,1\}$ & If $u_i$ is associated with $b_j$, $x_{ij}=1$.\\\hline
    $y_{ik}\in \{0,1\}$ & If $f_k$ is in cache of $b_j$, $y_{ik}=1$.\\
      \hline
   \end{tabular}\label{notation}
\end{table}

\subsection{Power Consumption at SBS}

Considering caching and backhaul power consumptions, total power consumed at SBS $b_{j}$ can be modeled: $p_{j}^{tot} = \rho p_{j}^{tr}+p_{j}^{hc},$
where $p_{j}^{tr}$ denotes transmit power consumed at $b_{j}$\cite{PowerPara1,PowerPara2}. $\rho$ reflects the impact of power amplifier and cooling on transmit power. $p_{j}^{hc}= p_{j}^{ca}+p_{j}^{cc}+p_{j}^{bh}$ is hardware and circuits-related power consumed at $b_j$, including power consumption for caching ($p_{j}^{ca}$), power consumption for operating  baseband and radio  circuits ($p_{j}^{cc}$) and power consumption for backhaul ($p_{j}^{bh}$). Usually, $p_{j}^{cc}$ is fixed. To quantify  power consumption for caching and backhaul, we adopt a power-proportional model\cite{PowerPara1,PowerPara2}. Specifically, $p_{j}^{ca}$ is proportional to the size of cached files at each SBS. Similarly, $p_{j}^{bh}$ is determined by the data rate of the backhaul link.

Let binary variable $y_{jk}$ indicate whether file $f_k$ is cached at $b_j$ or not. When file $f_k$ is cached at $b_j$, $y_{jk}=1$. Otherwise, $y_{jk} = 0$. A user's association policy is denoted by a binary variable $x_{ij}\in\{0, 1\}(\forall i \in \mathcal{U}, j\in \mathcal{B})$. If user $u_i$ is associated with SBS $b_j$, $x_{ij} = 1$. Otherwise, $x_{ij} = 0$. According to the power-proportional model, power consumption for caching at $b_j$ can be expressed as $p_{j}^{ca}= \sigma_{j}^{ca}\sum_{k\in \mathcal{F}} y_{jk}s_k$, where $\sigma_{j}^{ca}$ is the power coefficient of cache hardware in watt/bit and $\sum_{k\in \mathcal{F}} y_{jk}s_k$ is the size of cached files at $b_j$ in bits\cite{CachePower1}. power consumption for backhaul at $b_j$ is $p_{j}^{bh}=\sigma_{j}^{bh}\sum_{i\in \mathcal{U}}x_{ij}r_{ij}^{bh}$, where $\sigma_{j}^{bh}$ is the power coefficient of backhaul link, $r_{ij}^{bh}=\sum_{k\in\mathcal{F}}q_{ik}r_{k}(1-y_{jk})$ is the expected backhaul rate of uncached files for $u_i$ associated with $b_j$, $q_{ik}(k\in \mathcal{F})$ is the preference of $u_i$ for $f_{k}$ and $\sum_{k\in\mathcal{F}} q_{ik}= 1$\cite{CachePower1}. Finally, total power consumption at SBS $b_j$ can be obtained as
\begin{align*}
p^{tot}_{j}&=\rho\sum_{i\in\mathcal{U}}p_{ij}^{tr}+\sigma_{j}^{ca}\sum_{k\in\mathcal{F}}s_{k}y_{ik}+p_{j}^{cc}+\sigma_{j}^{bh}\sum_{i\in \mathcal{U}}x_{ij}r_{ij}^{bh}
\end{align*}

Orthogonal frequency division multiplexing (OFDM) is assumed to be used in the cache- enabled DSCNs. In this case, only inter-cell interference should be considered. For user $u_i$ associated with SBS $b_j$, interference from neighboring cells is denoted by $\sum\limits_{m\in\mathcal{U}}^{m\neq i}\sum\limits_{l\in\mathcal{B}}^{l\neq j} p_{ml}^{tr}g_{il}$. Then, signal-interference-noise-ratio (SINR) at user $u_i$ can be expressed as
\begin{equation*}
  \gamma_{ij} = \frac{p_{ij}^{tr}g_{ij}}{\sum\limits_{m\neq i}^{m \in \mathcal{U}}\sum\limits_{l\neq j}^{l\in\mathcal{B}} p_{ml}^{tr}g_{il}+N_0^2},
\end{equation*}
and the downlink data rate (bit/s) of $u_i$ can be derived as
\begin{equation}\label{DataRate}
  r_{ij}=Wlog_{2}(1+\gamma_{ij}),
\end{equation}
where $p_{ml}^{tr}$ is transmit power of $b_l$ for $u_m$, $g_{il}$ denotes the channel gain between $u_{i}$ and $b_{l}$, $N_0^2$ represents additive white Gaussian noise (AWGN) power, and $W$ indicates bandwidth allocated to each user.

\subsection{File Delivery Delay}
Part of files are cached at SBSs according to the file placement policy.  Let the file placement policy denoted by a binary variable $y_{jk}\in\{0,1\}(\forall j\in \mathcal{B},k\in \mathcal{F})$. Due to the limited cache capacity, $\sum_{k\in\mathcal{F}}y_{jk}s_k\leq M_j$ where $M_j (j\in\mathcal{B})$ is the cache capacity of SBS $b_j$.

In cache-enabled DSCN, file delivery delay consists of wireless transmission delay and backhaul delay. Consider a file delivery case that $u_i$ associated with SBS $b_j$ requests file $f_k$. As $f_k$ is delivered to $u_i$ through $b_j$, wireless transmission delay for $f_k$ can be derived as $\tau_{ijk}^{1}=\frac{s_{k}}{r_{ij}}$. According to \cite{backhauldelay}, delay of a backhaul link can be modeled as an exponentially distributed random variable with a mean value of $D_B$. With caching capabilities at SBSs, backhaul delay considered in file delivery delay depends on the file placement policy. To reflect this fact, backhaul delay for delivering $f_k$ through $b_j$ can be expressed as $\tau_{jk}^{2}= w_j^{bh}(1-y_{jk})$\cite{D-F}, where $w_j^{bh}$ denotes delay of $b_j$'s backhaul link. Then, file delivery delay for $f_k$ can be obtained as
\begin{equation}\label{delay function}
  d_{ij}^{k}=\tau_{ijk}^{1}+\tau_{jk}^{2}=\frac{s_{k}}{r_{ij}}+w_j^{bh}(1-y_{jk})
\end{equation}

Considering the file preference and association policy of $u_i$, we can finally derive average file delivery delay for $u_i$ as follows.
\begin{align*}
   d_{i}&=\sum\limits_{k\in\mathcal{F}}q_{ik}\sum\limits_{j\in\mathcal{B}}x_{ij}d_{ij}^{k}
\end{align*}

\section{Problem Formulation}
In this section, we firstly formulate the JDPO problem where power control, user association and file placement are jointly considered. Then, we decompose the JDPO problem into two subproblems, i.e., TPA problem related to power control and FDP problem related to user association and file placement. By relaxing the non-convex TPA problem to a convex one with the tight approximation, the JDPO problem is converted to a form that can be handled by GBD.

\subsection{Formulation of Joint File Delivery Delay and Power Consumption}
To represent the tradeoff relationship between file delivery delay and power consumption, a weighted sum utility function is used. Then, we can formulate the JDPO problem as a utility maximization problem. This method is widely used in multi-objective problem optimization, for example, in \cite{WeightedSum,WeightedSum1}. By jointly considering power control, user association and file placement, the JDPO problem is formulated as a mixed-integer programming (MIP) expressed as follows.
\begin{align}\label{tradeoff1}
\mathcal{P}1:&\min_{\bm{x},\bm{y},\bm{p}}
~~F(\bm{x},\bm{y},\bm{p})=\theta\delta_{p}\sum_{j\in \mathcal{B}} p^{tot}_{j}+ (1-\theta)\delta_{d}\sum_{i\in \mathcal{U}} d_i \\
\text{s.t.}~
&\sum_{i\in\mathcal{U}}x_{ij}p_{ij}^{tr}\leq P_{j}^{tr,max},~\forall j\in \mathcal{B}\tag{3-a}\label{cons:PowerConstraint},\\
&\sum\limits_{j \in \mathcal{B}}x_{ij}r_{ij}\geq R_i,~\forall i\in \mathcal{U}, \tag{3-b}\label{cons:FileDatarate}\\
&x_{ij}\in \{0,1\},~\forall i\in \mathcal{U},~\forall j\in\mathcal{B},\tag{3-c}\label{cons:Association}\\
&\sum_{j\in \mathcal{B}} x_{ij}= 1,~\forall i\in \mathcal{U},\tag{3-d}\label{cons:UniqueAssociation}\\
&y_{jk}\in\{0,1\},~\forall j\in \mathcal{B},~\forall k\in\mathcal{F},\tag{3-e}\label{cons:FilePlace}\\
&\sum_{k\in F}s_{k}y_{jk}\leq M_j,~\forall j\in\mathcal{B},\tag{3-f}\label{cons:CachingCapacity}\\
&\sum_{i\in\mathcal{U}}x_{ij}r_{ij}^{bh}\leq C_{j},\forall j\in \mathcal{B},\tag{3-g}\label{cons:BackhaulCapacity}
\end{align}
where $\bm{p}_{1\times UB}$, $\bm{x}_{1\times U B}$ and $\bm{y}_{1 \times BF}$ are power control, user association and file placement variable vectors, respectively. A larger balancing factor $\theta$$\in$$[0,1]$ means that less power consumption is preferred, however, at the expense of file delivery delay. $\delta_{p}$ and $\delta_{d}$ are normalization factors ensuring the same range for two objective functions\cite{Normalization}.  Constraint (\ref{cons:PowerConstraint}) requires  total transmit power should not exceed maximum available power at each SBS. Constraint (\ref{cons:FileDatarate}) represents the wireless transmission rate condition for user $u_i$ where $R_i = \sum\limits_{k\in \mathcal{F}}q_{ik}r_{k}$ denotes the average data requirement of  $u_i$. Each user association decision is indicated by a binary variable $x_{ij}$ and each user can at most be associated with one SBS, which are expressed in (\ref{cons:Association}) and (\ref{cons:UniqueAssociation}). In (\ref{cons:FilePlace}), each file placement decision is indicated by a binary variable $y_{jk}$.  The total size of files cached at an SBS can not exceed   maximum cache capacity of that SBS in (\ref{cons:CachingCapacity}).  The total file delivery data rate of a backhaul link should not exceed its maximum backhaul capacity in (\ref{cons:BackhaulCapacity}).

\newcounter{TempEqCnt} 
\setcounter{TempEqCnt}{\value{equation}} 
\setcounter{equation}{3} 
\begin{figure*}[b!]
\hrulefill
\begin{footnotesize}
\begin{align}
&F_{1}(\bm{p})
=\theta\delta_{p}\sum_{j\in \mathcal{B}}\sum_{i\in\mathcal{U}}\rho p_{ij}^{tr}
+(1-\theta)\delta_{d}\sum_{i\in\mathcal{U}}\sum\limits_{j\in \mathcal{B}}x_{ij}\sum\limits_{k\in \mathcal{F}}q_{ik}\frac{s_{k}}{Wlog_{2}{(1+\frac{p_{ij}^{tr}g_{ij}}{\sum\limits_{m\neq i}^{m \in \mathcal{U}}\sum\limits_{l\neq j}^{l\in\mathcal{B}} p_{ml}^{tr}g_{il}+N_0^2})}}\\
&F_{2}(\bm{x},\bm{y})
=\theta\delta_{p}\sum_{j\in\mathcal{B}}[\sigma_{j}^{ca}\sum_{k\in\mathcal{F}}s_{k}y_{jk}+p_{j}^{cc}+
\sigma_{j}^{bh}\sum_{i\in\mathcal{U}}x_{ij}\sum_{k\in\mathcal{F}}q_{ik}r_k(1-y_{jk})]\nonumber
+(1-\theta)\delta_{d}\sum_{i\in \mathcal{U}}\sum_{k\in \mathcal{F}} q_{ik}\sum_{j\in \mathcal{B}}x_{ij}w_{ij}^{bh}(1-y_{jk})\nonumber\\
\end{align}
\end{footnotesize}
\end{figure*}
\setcounter{equation}{\value{TempEqCnt}}

\textbf{\emph{Remark:}} $\mathcal{P}1$ is difficult to be solved directly due to the complex coupling relationship between file delivery delay and power consumption. We attempt to reduce the complexity of $\mathcal{P}1$ by decomposition. Two questions should be answered when solving the JDPO problem. The first question is, for each SBS, which power level should be employed for transmitting a requested file. The second question is, for each user, where to access its requested file, i.e., through which SBS? and then from cache or backhaul of the SBS? Based on these two questions, we decompose the JDPO problem into two subproblems, i.e., TPA problem to reflect the first question and FDP problem to reflect the second question. The TPA problem is related to power control while the FDP is related to user association and file placement. To realize decomposition, we rewrite the objective of $\mathcal{P}1$: $F(\bm{x},\bm{y},\bm{p})=F_{1}(\bm{p})+F_{2}(\bm{x},\bm{y})$, where $F_{1}{(\bm{p})},F_{2}(\bm{x},\bm{y})$ are shown at the bottom of this page. $F_{1}(\bm{p})$ is an objective of the TPA problem, which is related to continuous power control $\bm{p}$. $F_{2}(\bm{x},\bm{y})$ is an objective of the FDP problem, which is related to the binary user association policy $\bm{x}$ and file placement policy $\bm{y}$.

\subsection{Approximation of $F_1{(\bm{p})}$ }
Due to the non-convexity of the wireless transmission delay $\tau_{ijk}^{1}$, $F_1{(\bm{p})}$  in objective function $F(\bm{x},\bm{y},\bm{p})$ is   non-convex. To tackle the non-convexity of $F_1{(\bm{p})}$, an approximation relaxation method is considered.  This approximation is proved to be tight and have low computational complexity \cite{Appro1}\cite{Appro2}. The detailed two steps in the approximation are described as follows.

\setcounter{equation}{5}
\textbf{Step 1:} We will make use of the following lower bound:
      \begin{equation}\label{appro}
        \alpha log\gamma+\beta\leq log(1+\gamma)
\end{equation}
that achieves a tight result at $\gamma=\gamma_{0}$ when the approximation constants are
chosen as
\begin{align*}
\alpha   & = \frac{\gamma_0}{1+\gamma_0},~~~\beta  = log(1+\gamma_0)-\alpha log\gamma_0
\end{align*}

Based on the approximation, we can get:
\begin{align*}
r_{ij}=W(\alpha_{ij}log_{2}{(\frac{p_{ij}^{tr}g_{ij}}{\sum\limits_{m\neq i}^{m \in \mathcal{U}}\sum\limits_{l\neq j}^{l\in\mathcal{B}} p_{ml}^{tr}g_{il}+N_0^2})}+\beta_{ij})
\end{align*}
where the approximation parameters $\alpha_{ij}$ and $\beta_{ij}$ are fixed for each pair $<u_i,b_j> (i\in \mathcal{U},j\in \mathcal{B})$, obtained by the method in \cite{Appro2}.

 \textbf{Step 2:} We intend to use a log form such as $\widetilde{p_{ij}^{tr}}=log~ p_{ij}^{tr}~(~p_{ij}^{tr}=\exp(\widetilde{p_{ij}^{tr}}))$ to replace $p_{ij}^{tr}(i\in \mathcal{U},j\in \mathcal{B})$. Then, we have
\begin{align}\label{ConcaveRate}
\widetilde{r_{ij}}
= W(\alpha_{ij}log_{2}{(\frac{\exp(\widetilde{p_{ij}^{tr}})g_{ij}}{\sum\limits_{m\neq i}^{m \in \mathcal{U}}\sum\limits_{l\neq j}^{l\in\mathcal{B}} \exp(\widetilde{p_{ij}^{tr}})g_{il}+N_0^2})}+\beta_{ij})
\end{align}
which is a concave function over $\widetilde{p_{ij}^{tr}}$ \cite{Appro2}. Thus, according to the convexity rules, wireless transmission delay  $\tau_{ijk}^{1}=\frac{s_{k}}{\widetilde{r}_{ij}}$ becomes convex \cite{Convex}.

After the above two steps, non-convex $F_1(\bm{p})$ is replaced approximatively by a convex $\widetilde{F}_{1}(\widetilde{\bm{p}})$ where $\widetilde{\bm{p}} =(\widetilde{p_{ij}^{tr}},i\in \mathcal{U},j\in \mathcal{B})$. $\widetilde{F}_{1}(\widetilde{\bm{p}})$ is expressed at the bottom of this page.

\newcounter{TempEqCnt1} 
\setcounter{TempEqCnt1}{\value{equation}} 
\setcounter{equation}{7} 
\begin{figure*}[b!]
\hrulefill
\begin{footnotesize}
\begin{align}
\widetilde{F}_{1}(\widetilde{\bm{P}})
&=\theta\delta_{p}\sum_{j\in \mathcal{B}}\sum_{i\in\mathcal{U}}\rho \exp(\widetilde{p_{ij}^{tr}})
+(1-\theta)\delta_{d}\sum_{i\in\mathcal{U}}\sum\limits_{j\in \mathcal{B}}x_{ij}\sum_{k\in \mathcal{F}} q_{ik}\frac{s_{k}}{W[\alpha_{ij}log_{2}{(\frac{\exp(\widetilde{p_{ij}^{tr}})g_{ij}}{\sum\limits_{m\neq i}^{m \in \mathcal{U}}\sum\limits_{l\neq j}^{l\in\mathcal{B}} \exp(\widetilde{p_{ml}^{tr}})g_{il}+N_0^2})}+\beta_{ij}]}\\
=\theta\delta_{p}\sum_{j\in \mathcal{B}}&\sum_{i\in\mathcal{U}}\rho \exp(\widetilde{p_{ij}^{tr}})
+(1-\theta)\delta_{d}\sum_{i\in\mathcal{U}}\sum\limits_{j\in \mathcal{B}}x_{ij}\sum_{k\in \mathcal{F}} \frac{q_{ik}s_{k}}{W[\alpha_{ij}(1.44\widetilde{p_{ij}^{tr}}+log_{2}(g_{ij})-log_{2}(\sum\limits_{m\neq i}^{m \in \mathcal{U}}\sum\limits_{l\neq j}^{l\in\mathcal{B}} \exp(\widetilde{p_{ml}^{tr}})g_{il}+N_0^2))+\beta_{ij}]}\nonumber
\end{align}
\end{footnotesize}
\end{figure*}
\setcounter{equation}{\value{TempEqCnt1}} 

\subsection{Reformulation of Problem $\mathcal{P}1$}
By relaxing the non-convex TPA problem to a convex one with tight approximation, $\mathcal{P}$1 can be reformulated into $\mathcal{P}1'$ with $\widetilde{F}_{1}(\bm{\widetilde{p}})$ and $F_2(\bm{x},\bm{y})$ as follows.
\setcounter{equation}{8}
\begin{align}
\mathcal{P}1':
&\min_{\bm{x},\bm{y},\bm{\widetilde{p}}}~~\widetilde{F}(\bm{x},\bm{y},\bm{\widetilde{p}})=\widetilde{F}_{1}(\bm{\widetilde{p}})+F_{2}(\bm{x},\bm{y})\\
&~~\text{s.t.}~~~~~(\ref{cons:PowerConstraint})\thicksim(\ref{cons:BackhaulCapacity})~\text{with}~\bm{\widetilde{p}}~\text{instead}\nonumber
\end{align}


We can see that $\mathcal{P}1'$ is still an MIP problem. Although the complexity of $F_1(\bm{p})$ is reduced by the convex approximation, $\mathcal{P}1'$ is still  NP-hard  with exponential computation time \cite{NPHard}. In the next section, we find $\mathcal{P}1'$ can be handled by GBD.


\section{Problem Analysis}
In this section, we first analyze the structure of $\mathcal{P}1'$ and confirm that it can be handled by GBD. Then we decompose $\mathcal{P}1'$ into two smaller problem, i.e., \emph{primal problem} related to power control and \emph{master problem} related to user association and file placement. By solving the primal problem and master problem iteratively, a sequence of non-increasing upper bounds as well as no-decreasing lower bounds can be obtained to approach the optimal solution of $\mathcal{P}1'$.

\subsection{GBD Approach and Problem $\mathcal{P}1'$}
GBD is a powerful approach for solving a certain kind of MIP problems [R]. The basic idea of GBD is to decompose the original MIP problem into a \emph{primal problem} and a \emph{master problem}, then solve these two smaller problems iteratively. The primal problem is a convex programming problem and its solution results in a upper bound of the original problem. Then with the solution of the primal problem, the remaining master problem is solved to get a lower bound of the original problem. The GBD approach uses a sequence of non-increasing upper bounds and non-decreasing lower bounds to approach the optimal solution.

In $\mathcal{P}1'$, $\widetilde{F}_{1}(\bm{\widetilde{p}})$ corresponds to the convex approximation of the TPA problem determined by continuous power control $\bm{\widetilde{p}}$, while $F_{2}(\bm{x},\bm{y})$ corresponds to the FDP problem determined by binary user association policy $\bm{x}$ and file placement policy $\bm{y}$. Moreover, constraints (\ref{cons:PowerConstraint}) and (\ref{cons:FileDatarate}) are transmit power $\bm{\widetilde{p}}$-related inequations. And the constraints  (\ref{cons:Association})-(\ref{cons:BackhaulCapacity}) are only related to binary variables $\bm{x}$ and $\bm{y}$. Only the constraint (\ref{cons:PowerConstraint}) contains continue and binary variables $(\widetilde{\bm{p}},\bm{x})$. Due to these separation features in both objective and constraints, the GBD approach can be applied to solve $\mathcal{P}1'$ \cite{GBD}.

\subsection{Primal Problem}
According to the GBD approach, we first fix the binary variables $(\bm{x}^{(t)},\bm{y}^{(t)})$ in $\mathcal{P}1'$ at $t$-th iteration, and we can obtain the primal problem with  only continuous variables $\bm{p}^{(t)}$, namely convex approximation of the \textbf{TPA} problem:
\begin{align}
\mathcal{P}2:\min_{\bm{\widetilde{p}}} &~~\widetilde{F}_1(\bm{\widetilde{p}}) \label{trpower}\\
\text{s.t.}
& 0\leq \exp({\widetilde{p_{ij}^{tr}}})\leq x_{ij}^{(t)}P_{j}^{tr,max},~\forall i\in \mathcal{U},~\forall j\in \mathcal{B},\tag{10-a}\label{cons:P2-1}\\
&0\leq \sum_{i\in\mathcal{U}}\exp({\widetilde{p_{ij}^{tr}}})\leq P_{j}^{tr,max},~\forall j\in \mathcal{B},\tag{10-b}\label{cons:P2-2}\\
& \sum\limits_{i \in \mathcal{U}}\sum\limits_{j \in \mathcal{B}}x_{ij}^{(t)}\widetilde{r_{ij}} \geq \sum\limits_{k\in \mathcal{F}}q_{ik}r_{k},\tag{10-c}\label{cons:P2-3}
\end{align}
where (\ref{cons:PowerConstraint}) is divided equivalently into (\ref{cons:P2-1}) and (\ref{cons:P2-2}) to make the following feasibility discussion simpler. According to Sec.III, $\widetilde{F}_1({\bm{\widetilde{p}}})$ is convex and $\widetilde{r}_{ij}$ of (\ref{cons:P2-3}) is concave.  (\ref{cons:P2-1}) and (\ref{cons:P2-2}) are also convex constraints due to the convexity of exponential function. Therefore, we can conclude that $\mathcal{P}2$ is a convex problem.

\subsection{Feasibility Discussion of $\mathcal{P}2$}
Given variables $(\bm{x}^{(t)},\bm{y}^{(t)})$ will affect the feasibility of primal problem $P$2. Therefore, before presenting the master problem, we will discuss the feasibility of $\mathcal{P}2$.

\textbf{Feasible Case}: For given variables $(\bm{x}^{(t)},\bm{y}^{(t)})$ , if primal problem $\mathcal{P}2$ is feasible, it is easy to obtain the solution $\bm{p}^{(t)}$ . Then, according to GBD, the dual problem of $\mathcal{P}2$ should be analyzed to formulate the master problem. We define the partial Lagrangian function of $\mathcal{P}2$ as
\begin{align}\label{DualPrimalFun}
&\mathcal{L}(\bm{x}^{(t)},\bm{\mu},\bm{\widetilde{p}})\nonumber\\
&= \widetilde{F}_1(\bm{\widetilde{p}})+\sum_{i\in \mathcal{U}}\sum_{j\in \mathcal{B}}\mu_{ij}(\exp{(\widetilde{p_{ij}^{tr}})}-x_{ij}^{(t)}P_{j}^{tr,max}),
\end{align}
where multipliers $\bm{\mu}$ corresponding to constraints (\ref{cons:P2-1}) should satisfy $\mu_{ij}\geq 0,(\forall i\in \mathcal{U}, \forall j\in \mathcal{B}$). The dual problem of $\mathcal{P}2$ is described as follows.
\begin{align}\label{DualPrimalProblem}
&\max_{\bm{\mu}}~\min_{\bm{\widetilde{p}}}\mathcal{L}(\bm{x}^{(t)},\bm{\mu},\bm{\widetilde{p}})\\
&~~\text{s.t.}~~~\mu_{ij}\geq 0,\forall i\in \mathcal{U},~\forall j\in \mathcal{B}\tag{12-a}\nonumber\\
&~~~~~~~~(\ref{cons:P2-1})\thicksim(\ref{cons:P2-3})\nonumber
\end{align}

By solving primal problem $\mathcal{P}2$ and its dual problem (\ref{DualPrimalProblem}), we can get optimal power solution $\bm{p}^{(t)}$ and dual solution $\bm{\mu}^{(t)}$. Then, both $\bm{p}^{(t)}$ and $\bm{\mu}^{(t)}$ will be used as  known conditions passed to the master problem.

\textbf{Infeasible  Case}:
If primal problem $\mathcal{P}2$ is infeasible for given binary variables $(\bm{x}^{(t)},\bm{y}^{(t)})$, we first need to identify the infeasible constraints in (\ref{cons:P2-1}) of $\mathcal{P}2$. Referring to \cite{ConsVio}, in Proposition \ref{Infea}, we introduce a constraint violation problem (V) to locate the infeasible constraints in (\ref{cons:P2-1}).


\begin{proposition}\label{Infea}
First, we focus on the constraints (\ref{cons:P2-1}) and a constraint violation problem (V) is defined as follows.
\begin{align*}
\tag{V}
  &\min\limits_{\eta,\bm{\widetilde{p}}}\quad\eta \\
  &~\text{s.t.}~\exp({\widetilde{p_{ij}^{tr}}})- x_{ij}^{(t)}P_{j}^{tr,max}\leq\eta,~\forall i\in \mathcal{U},~\forall j\in \mathcal{B},\\
  &\quad\quad~~\eta\geq0,\\
  &\quad\quad~~(\ref{cons:P2-2})\thicksim(\ref{cons:P2-3}),
\end{align*}
The dual problem of V is
\begin{align}\tag{V-Dual}
  &\min\limits_{\eta,\bm{\nu},\bm{\widetilde{p}}}\quad\overline{\mathcal{L}}(\bm{\bm{x}^{(t)},\bm{\nu},\widetilde{p}})\\ & =\eta  \nonumber
   +\sum_{i\in \mathcal{U}}\sum_{j\in \mathcal{B}}\nu_{ij}(\exp{(\widetilde{p}_{ij})}-x_{ij}^{(t)}P_{j}^{tr,max}-\eta) \\
  &~\text{s.t.}\quad~~\eta\geq0,\nonumber\\
  &\quad\quad~~(\ref{cons:P2-2})\thicksim(\ref{cons:P2-3}),\nonumber
\end{align}
where $(\bm{\widetilde{p}},\eta)$ and $\bm{\nu}$ are the variables and dual variables for the convex feasible problem $(V^{(t)})$ at $t$-th iteration.
 After solving  the dual problem of V, the optimal dual solutions $\bm{\nu}^{(t)}$ can locate the infeasible constraints in (\ref{cons:P2-1}).
\end{proposition}
\begin{proof}
In problem ($V$), as all constraints and the objective $\eta$ are convex ,  ($V$) is a convex problem.
Partial Langragian function $\overline{\mathcal{L}}$ of problem ($V$) can be obtained as
\begin{align}\label{Lv}
  &\overline{\mathcal{L}}(\bm{x}^{(t)},\bm{\nu},\bm{\widetilde{p}})\nonumber\\
  & =\eta
   +\sum_{i\in \mathcal{U}}\sum_{j\in \mathcal{B}}\nu_{ij}(\exp{(\widetilde{p}_{ij})}-x_{ij}^{(t)}P_{j}^{tr,max}-\eta)
\end{align}

By introducing the dual problem and dual variables $\bm{\nu}$, we can see all constraints in (\ref{cons:P2-1}) are coupled with $\bm{\nu}$.
Let $(\bm{\widetilde{p}}^{t},\eta^{(t)})$ and $\bm{\nu}^{(t)}$ be the optimal solution and dual solution, respectively. Then we have
\begin{align}\label{LDualF}
  (\bm{\widetilde{p}}^{t},\eta^{(t)})
  &=\mathop{\argmin}_{\bm{\widetilde{p}},\eta\geq0}~~  \overline{\mathcal{L}}(\bm{x}^{(t)},\bm{\nu}^{(t)},\bm{\widetilde{p}})\nonumber\\
  &=\mathop{\argmin}_{\bm{\widetilde{p}},\eta}~~
  (1-\sum_{i\in \mathcal{U}}\sum_{j\in \mathcal{B}}\nu_{ij})\eta\nonumber\\
  &+\sum_{i\in \mathcal{U}}\sum_{j\in \mathcal{B}}\nu_{ij}(\exp{(\widetilde{p}_{ij})}-x_{ij}^{(t)}P_{j}^{tr,max})
\end{align}
According to the convex dual theorem, the optimality condition is that $\frac{
\partial\overline{\mathcal{L}}(\bm{\widetilde{p}},\eta,\bm{\nu}^{(t)},\bm{x}^{(t)})}{
\partial\eta}=0$. Thus, $\bm{\nu}$ must satisfy: $\sum_{i\in \mathcal{U}}\sum_{j\in \mathcal{B}}\nu_{ij}^{(t)}=1$. From (\ref{LDualF}), we have
\begin{align}\label{pt}
\bm{\widetilde{p}}^{(t)} = \mathop{\argmin}_{\bm{\widetilde{p}}}\sum_{i\in \mathcal{U}}\sum_{j\in \mathcal{B}}\nu_{ij}^{(t)}\exp{(\widetilde{p_{ij}})}
\end{align}
\end{proof}
Optimal solution $\bm{\widetilde{p}}^{(t)}$ and dual solution $\bm{\nu}^{(t)}$ will be used to form the master problem.
\subsection{Master Problem}
The master problem is used to determine the binary variables (i.e., user association and file placement policies), which corresponds to the \textbf{FDP} Problem. Primal problem $\mathcal{P}2$ is first solved and the solutions of $\mathcal{P}2$ (feasible case or infeasible case) will be used to construct the constraints of master problem.

According to the Theorem 2.2 in \cite{GBD}, the master problem can be expressed as follows.
\begin{align}\label{MasterProblem}
\mathcal{P}3:
  \min_{\phi,\bm{x},\bm{y}} &\quad\phi+F_2{(\bm{x},\bm{y})}\\
  \text{s.t.}~~
  &\quad \phi \geq   \mathcal{L}(\bm{x},\bm{\mu}_{fea}^{(t_1)},\bm{\widetilde{p}}_{fea}^{(t_1)}),t_1= 1,2... T_1  \tag{16-a}  \label{cons:Feasible}\\
  &\quad  0   \geq   \overline{\mathcal{L}}(\bm{x},\bm{\nu}_{inf} ^{(t_2)},\bm{\widetilde{p}}_{inf} ^{(t_2)}),t_2 = 1,2... T_2             \tag{16-b}  \label{cons:InFeasible}\\
  &\quad(\text{16-c})\thicksim(\text{16-g})\nonumber
\end{align}
where constraints $(\text{16-c})\thicksim(\text{16-g})$ are the same as  constraints $  (\text{3-c})\thicksim (\text{3-f})$ in problem $\mathcal{P}1$.                        constraints (\ref{cons:Feasible}) and (\ref{cons:InFeasible}) are the optimal cut constraints and feasible cut constraints, respectively, according to GBD. $ \mathcal{L}(\bm{x},\bm{\mu}_{fea}^{(t_1)},\bm{\widetilde{p}}_{fea}^{(t_1)})= \widetilde{F}_1(\bm{\widetilde{p}})+\sum_{i\in \mathcal{U}}\sum_{j\in \mathcal{B}}\mu_{ij}(\exp{(\widetilde{p_{ij}^{tr}})}-x_{ij}^{(t)}P_{j}^{tr,max}) $ and $\overline{\mathcal{L}}(\bm{x},\bm{\nu}_{inf} ^{(t_2)},\bm{\widetilde{p}}_{inf} ^{(t_2)})= \sum_{i\in \mathcal{U}}\sum_{j\in \mathcal{B}}\nu_{ij}(\exp{(\widetilde{p}_{ij})}-x_{ij}^{(t)}P_{j}^{tr,max}) $. $\bm{\widetilde{p}}_{fea}^{(t_1)}$ and $\bm{\mu}_{fea}^{(t_1)}$ are the optimal and dual solutions of problem $\mathcal{P}2$ if $\mathcal{P}2$ is feasible at $t$-th iteration. Besides,  $\bm{\widetilde{p}}_{inf}^{(t_2)}$ and $\bm{\nu}_{inf}^{(t_2)}$ are the optimal and dual solutions of problem $V$ when $\mathcal{P}2$ is infeasible at $t$-th iteration. Index $t_1$ and $t_2$ record the $t_1$-th feasible and $t_2$-th infeasible problem  $\mathcal{P}2$, respectively. $T_1$ and $T_2$ denote the number of the feasible and infeasible problem $\mathcal{P}2$, respectively. Apparently, at $t$-th iteration, $T_1+T_2 = t$.

According to the GBD approach, during each iteration, an upper bound and a lower bound of the problem $\mathcal{P}1'$ can be obtained, which are described in Proposition \ref{bound}.

\begin{proposition}\label{bound}
At each iteration, an upper bound and a lower bound of problem $\mathcal{P}1'$ are obtained. At the $t$-th iteration, lower bound $LB^{(t)}=N^{(t)}$ where $N^{(t)}$ is the optimal objective value of master problem $\mathcal{P}$3 at $t$-th iteration. Upper bound $UB^{(t)}=\min\limits_{0\leq r\leq t} {(M^{(r)}+F_2{(\bm{x}^{(r)},\bm{y}^{(r)})}}$ where $M^{(r)}=\max\limits_{\bm{\mu}}\min\limits_{\bm{\widetilde{p}}}\mathcal{L}(\bm{x}^{(r)},\bm{\mu},\bm{\widetilde{p}})$ is the optimal objective value of $\mathcal{P}2$'s dual problem at $r$-th iteration.
\end{proposition}

\begin{proof}
\textbf{Lower Bound:}
We first consider how to calculate lower bound $LB^{(t)}$ of problem $\mathcal{P}1'$ at the $t$-th iteration. The lower bound of objective value in problem $\mathcal{P}1'$ is mainly obtained by solving master problem $\mathcal{P}$3.

The objective of problem $\mathcal{P}1'$ is $\min\limits_{\bm{x},\bm{y},\bm{\widetilde{p}}}~~\widetilde{F}(\bm{x},\bm{y},\bm{\widetilde{p}})=\widetilde{F}_{1}(\bm{\widetilde{p}})+F_{2}(\bm{x},\bm{y})$.
 The  problem $\min\limits_{\bm{\widetilde{p}}} \widetilde{F}_1(\bm{\widetilde{p}})$ is  convex. Hence, according to the strong duality of the  convex problem, is equivalent to $\max\limits_{\bm{\mu}} \min\limits_{\bm{\widetilde{p}}}\mathcal{L}(\bm{x}^{(t)},\bm{\mu},\bm{\widetilde{p}})$. Therefore, the original problem $P1'$ is equivalent to
\begin{align}\label{BoundLowerProof}
  \min_{\bm{u},\bm{v},\bm{x},\bm{y}} &\quad \min_{\bm{\widetilde{p}}}\mathcal{L}(\bm{x},\bm{\mu} ,\bm{\widetilde{p}})+ F_2 (\bm{x},\bm{y}) \\
    s.t. &\quad  (\ref{cons:Association})\thicksim (\ref{cons:BackhaulCapacity})\nonumber
\end{align}

Comparing problem $\mathcal{P}3$ with problem (\ref{BoundLowerProof}), we can see that problem $\mathcal{P}$3 with constraints (\ref{cons:Feasible}) and (\ref{cons:InFeasible}) is the relaxation of problem (\ref{BoundLowerProof}). Therefore, the solution search space of problem $\mathcal{P}$3 is larger than that in problem (\ref{BoundLowerProof}). which makes $N^{(t)}$ smaller than the optimal objective value of problem (\ref{BoundLowerProof}). According to the duality theory, problem $\mathcal{P}$1 and its dual problem (\ref{BoundLowerProof}) have the same optimal objective value. Hence $N^{(t)}$ is smaller than the optimal objective value of  problem $\mathcal{P}$1.

we can conclude that lower bound  $LB^{(t)}$ of the optimal objective value of problem $\mathcal{P}$1 is $N^{(t)} =\phi^{(t)}+F_2(\bm{x}^{(t)},\bm{y}^{(t)})$ at $t$-th iteration.

\textbf{Upper Bound:}
We prove that the upper bound of $\mathcal{P}1'$ is $UB^{(t)}=\min\limits_{0\leq r\leq t} {(M^{(r)}+F_2{(\bm{\widetilde{x}}^{(r)},\bm{\widetilde{y}}^{(r)})}}$ where $M^{(r)}=\max\limits_{\bm{\mu}}\min\limits_{\bm{\widetilde{p}}}\mathcal{L}(\bm{x}^{(r)},\bm{\mu},\bm{\widetilde{p}}) (0\leq r \leq t)$ is the optimal objective value of dual problem of $\mathcal{P}2$ at $r$-th iteration.

Back to problem $\mathcal{P}2$, as $\bm{x}^{(r)}$ will make problem $\mathcal{P}2$ either infeasible or feasible, the optimal objective value of $\mathcal{P}2$ either infinite or finite. According to the duality theory, the optimal objective value $M^{(r)}$ of dual problem of $\mathcal{P}2$ is larger than that of $\mathcal{P}2$. For the infeasible case,  $M^{(r)}$ is infinite.  It is apparent that $\min\limits_{0\leq r \leq t}\{M^{(r)}+F_2{(\bm{\widetilde{x}}^{(r)},\bm{\widetilde{y}}^{(r)})}\}$ is the upper bound.

For the feasible problem $\mathcal{P}2$, the contradiction method is used to prove the upper bound. We assume that $UB^{(t)} \leq G^{*}$ where $G^{*}$ denotes the optimal objective value of problem $\mathcal{P}1'$.  Among $t$ iterations,  $(M^{(\omega)}+F_2(\bm{\widetilde{x}}^{(\omega)},\bm{y}^{(\omega)})$ is the minimum value (i.e.  $\omega=\arg\min\limits_{0\leq r \leq t}\{M^{(r)}+F_2(\bm{\widetilde{x}}^{(\omega)},\bm{\widetilde{y}}^{(\omega)})$), where $0\leq\omega\leq t$. Let power $\bm{\widetilde{p}}^{(\omega)}$ and $(\bm{\mu}^{(\omega)},\bm{\nu}^{(\omega)})$ be the solutions to problem $\mathcal{P}2$ and  its dual problem at $\omega$-th iteration, respectively. According to the strong duality and the assumption, we have
\begin{align}\label{Uppper}
  UB^{(t)} &=\mathcal{L}(\bm{x}^{(\omega)},\bm{\mu}^{(\omega)},\bm{\widetilde{p}}^{(\omega)})+F_2(\bm{x}^{(\omega)},\bm{y}^{(\omega)})\nonumber\\  & =\widetilde{F}_1(\bm{\widetilde{p}}^{(\omega)})+F_2(\bm{x}^{(\omega)},\bm{y}^{(\omega)})\nonumber\\
  &\leq \widetilde{F}_1(\bm{\widetilde{p}}^{(*)})+F_2(\bm{x}^{(*)},\bm{y}^{(*)})=G^{*}
\end{align}
where $(\bm{\widetilde{p}}^{(*)},\bm{x}^{(*)},\bm{y}^{(*)})$ is the optimal solution to problem $\mathcal{P}1'$. Inequation (\ref{Uppper}) shows that there exist a smaller objective value of problem $\mathcal{P}1'$ given by solution $(\bm{\widetilde{p}}^{(\omega)},\bm{x}^{(\omega)},\bm{y}^{(\omega)})$ than that given by the optimal solution $(\bm{x}^{(*)},\bm{p}^{(*)})$. Such result is contradictory to the initiatory assumption that $(\bm{\widetilde{p}}^{(*)},\bm{x}^{(*)},\bm{y}^{(*)})$ is the optimal solution
to problem $\mathcal{P}1'$. Therefore, $\min\limits_{0\leq r\leq t} {(M^{(r)}+F_2{(\bm{\widetilde{x}}^{(r)},\bm{\widetilde{y}}^{(r)})}}$ is always larger than the optimal objective value of problem $\mathcal{P}$1 .

Hence, the upper bound of problem $\mathcal{P}$1 is $UB^{(t)}=\min\limits_{0\leq r\leq t} {(M^{(r)}+F_2{(\bm{\widetilde{x}}^{(r)},\bm{\widetilde{y}}^{(r)})}}$  .
\end{proof}

\section{Proposed Algorithm}
In this section, based on GBD, we propose an iterative algorithm to obtain an $\epsilon$-optimal objective value of problem $P1'$. To further reduce the computational complexity of the master problem $\mathcal{P}3$, the SDR technique is used and a corresponding accelerated algorithm is proposed.

\subsection{Algorithm Design}
According to the GBD approach, at each iteration, an upper bound and lower bound of the original problem can be obtained. The upper bound is decreasing while the lower bound is increasing after each iteration. When two bounds are closing to each other, the optimal solution and objective value can be approached. To solve $\mathcal{P}1'$, joint power control, user association and file placement (PUF) algorithm is proposed and the procedure is described in Algorithm \ref{alg1}.

First, iteration index $t$ is set to 0 and the initial assignment of $(\bm{x},\bm{y})$ are $(\bm{x}^{(0)},\bm{y}^{(0)})$. Then the primal problem $P$2 is solved with given $\bm{x}=\bm{x}^{(0)}$ where $\bm{x}^{(0)}$ can be randomly selected at first. When problem $P$2 is feasible, the optimal power solution $\bm{\widetilde{p}}^{(t)}$ and its optimal dual solutions $\bm{\mu}^{(t)}$ the new optimal cut $\phi \geq\mathcal{L}(\bm{x},\bm{\mu}^{(t)},\bm{\widetilde{p}}^{(t)})$ can be obtained. If the primal  problem $P$2 is infeasible, problem $V^{(t)}$ is formulated and solved with $\bm{x}=\bm{x}^{(0)}$. Solution $(\bm{\widetilde{p}}^{(t)},\bm{\nu}^{(t)})$ to problem $V^{(t)}$ are used to generate a new feasible cut $0\geq L(\bm{x},\bm{\nu}^{(t)},\bm{\widetilde{p}}^{(t)})$. Either optimal cut or feasible cut is added to the constraint set of master problem $\mathcal{P}$3. With the updated and accumulated constraints, problem $\mathcal{P}$3 is solved to produce solution ($\bm{x},\bm{y}$) for the next iteration.

During each iteration, the optimal value of problem $P$3 gives lower bound $LB^{(t)}$ of problem $\mathcal{P}$1. Such lower bound is non-decreasing with $t$ because at each iteration the feasible region of problem $\mathcal{P}$3 is shrunk with newly added constraint cuts. Besides, upper bound $UB^{(t)}$ of $\mathcal{P}$1 is obtained by solving feasible problem $\mathcal{P}2$. The upper bound is non-increasing with iteration index $t$. The repeated procedure of solving problem $P$2 and problem $P$3 will  terminate until $|UB^{(t)}-LB^{(t)}|\leq \epsilon$ where $\epsilon$ is sufficiently small.

\begin{algorithm}
  \caption{Power control, User association and File placement (PUF) Algorithm for $\mathcal{P}1'$}
   \label{alg1}
   \KwIn{Wireless network parameters $p_{j}^{max},M_j,C_{j}$, balancing and normalization parameters $~\theta,~\delta_{d},~\delta_{p},~\forall i\in\mathcal{U},~\forall j\in \mathcal{B}.$}
   \KwOut{Optimal power control$~\bm{p}^*$, user association$~\bm{x}^*$, file placement $~\bm{y}*$}
   \textbf{Initialization:}~Let$~\bm{x}=\bm{x^{(0)}},~LB^{(0)}=-\infty, ~UB^{(0)}=+\infty,~t=1,t_1=1,t_2=1$\;

   \Repeat{$~|UB^{(t-1)}-~LB^{(t-1)}|\leq \epsilon$}{
   \textbf{Find Solution of Primal Problem \bm{\mathcal{P}}2:}\\
   Solve $\mathcal{P}2$ with Interior Point Method\;
   \eIf{\rm{ $\mathcal{P}2$ is bounded}}{
   solve $\min\limits_{\bm{x}^{(t-1)}}\mathcal{L}(\bm{x},\bm{\mu},\bm{\widetilde{p}})$ to get $(\bm{\widetilde{p}}^{(t)},\bm{\mu}^{(t)})$ and generate the new optimal cut $\phi \geq   \mathcal{L}(\bm{x},\bm{\mu}^{(t)},\bm{\widetilde{p}}^{(t)})$;
   store $(\bm{\mu}^{(t)},\bm{\widetilde{p}}^{(t)})$ into $(\bm{\mu}_{fea}^{(t_1)},\bm{\widetilde{p}}_{fea}^{(t_1)})$;$t_1=t_1+1$;
   }
   {solve problem $V$ to get $(\bm{\widetilde{p}}^{(t)},\bm{\nu}^{(t)})$ and generate the new feasible cut $0 \geq   \mathcal{L}(\bm{x},\bm{\nu}^{(t)},\bm{\widetilde{p}}^{(t)})$;store $(\bm{\nu}^{(t)},\bm{\widetilde{p}}^{(t)})$ into $(\bm{\nu}_{inf}^{(t_2)},\bm{\widetilde{p}}_{inf}^{(t_2)})$;
   $t_2=t_2+1$;}
   Calculate upper bound $UB^{(t)}=\min\limits_{0\leq r\leq t}{(M^{(r)}+F_2{(\bm{\widetilde{x}}^{(r)},\bm{\widetilde{y}}^{(r)})})}$;

   \textbf{Find Solution of Master Problem \bm{\mathcal{P}}3:}\\
    Add a constraint: $\phi \geq   \mathcal{L}(\bm{x},\bm{\mu}^{(t)},\bm{\widetilde{p}}^{(t)})$ or $0 \geq   \mathcal{L}(\bm{x},\bm{\nu}^{(t)},\bm{\widetilde{p}}^{(t)})$ to $\mathcal{P}$3\;
   Solve $\mathcal{P}$3 to obtain $\bm{x}^{(t)}$\ and  $\bm{y}^{(t)}$;
   Calculate lower bound: $LB^{(t)}=\phi^{(t)}+F_2(\bm{x}^{(t)},\bm{y}^{(t)})$\;
   $t=t+1$\;
   }
   Get optimal solution $\bm{p}^{*}$,$\bm{x}^{*}$and $\bm{y}^{*}$.
   \end{algorithm}

Indeed, GBD-based PUF algorithm (Algorithm 1) is convergent after a finite number of iterations and an $\epsilon$-optimal objective value of problem $\mathcal{P}$1 can be obtained.

  \begin{proposition}\label{covergency}\textbf{Convergency Analysis of PUF:}
The PUF algorithm will obtain an $\epsilon$-optimal objective value of problem $\mathcal{P}1'$ after a finite number of iterations.
\end{proposition}\begin{proof}
To prove the convergency of PUF algorithm, the proof follows the result of  Theorem 2.4 of GBD \cite{GBD,GBD1,GBD2}.

Actually, due to the finiteness of discrete set $(\bm{x},\bm{y})$ in problem $\mathcal{P}$1, the number of iteration of PUF is finite. According to   $UB^{(t)}=\min\limits_{0\leq r\leq t} {(M^{(r)}+F_2{(\bm{\widetilde{x}}^{(r)},\bm{\widetilde{y}}^{(r)})}}$ of upper bound, upper bound $UB^{(t)}$ of problem $P1'$ is nonincreasing with iteration $t$.  The added constraints makes search space of $(\bm{x},\bm{y})$ in relaxed master problem $\mathcal{P}$3 smaller which implies lower bound $LB^{(t)}$  of  $\mathcal{P}$3 is nondecreasing after each iteration. Hence, the gap of the upper and lower bound is shrunk after each iteration. The PUF algorithm procedure terminates in a finite number of iterations when the gap of the upper and lower bound is less than $\epsilon$. Therefore, the proposed PUF algorithm can converge to a $\epsilon$-optimal objective value of problem $\mathcal{P}1'$.
\end{proof}

\textbf{Complexity Analysis of PUF :}
At each iteration, two subproblems-primal problem $\mathcal{P}2$ and master problem $\mathcal{P}$3 are solved alternately. During the solution procedure of $\mathcal{P}2$, Interior Point Method is used and the computational complexity is $O((UB)^{3})$ \cite{IPM}. However, to solve master problem $\mathcal{P}$3, all the possible binary feasible $(\bm{x},\bm{y})$ in the constraints need to be searched, which incurs an exponential computational complexity (from step (11) to step (13), Algorithm 1). Thus, we propose a fast and efficient SDR-based algorithm to find optimal user association and file placement policies.

\subsection{Accelerated Algorithm for Master Problem  $\mathcal{P}$3} Master problem $\mathcal{P}$3 belongs to a general quadratically constrained quadratic programming (QCQP) problem. The prevailing method to
tackle the general QCQP problem is through semi-definite relaxation (SDR). By SDR, some quantified sub-optimality can be  guaranteed \cite{SDR1,SDR2}. Then an SDR-based algorithm is proposed to solve $\mathcal{P}$3.

\subsubsection{General QCQP Problem and SDP Relaxation}
A general QCQP   problem can be expressed as follows\cite{SDR}.
\begin{align*}
  \min_{\bm{x}} &~~~\bm{x^{T}A_{0}x+b_{0}^{T}x}\\
  s.t.   &~~~\bm{x^{T}A_{i}x+b_{i}^{T}x}\preceq c_{i}, \forall i=1,...M,
\end{align*}
where $\bm{x}$ is a $1\times n$ variable vector. $\bm{b}_{i}$ and $c_{i}~(\forall i=1,...M)$ is a $1\times n$ constant vector. $\bm{A}_{i}~(\forall i=0,...,M)$ is a $n \times n$ coefficient matrix. The SDR technique is widely used to solve the non-convex QCQP problem. The procedure is described as follows.
\begin{align*}
  \min_{\bm{X},\bm{x}}    &~~~Trace\{\bm{A_{0}X}\}+\bm{b}_{0}^{T}\bm{x}\\
  \text{s.t.} ~
    &~~~Trace\{\bm{A_{i}X}\}+\bm{b}_{i}^{T}\bm{x} \preceq c_{i}, \forall i=1,...M,\\
    &~~~\left[ \begin{array}{cc}
             \bm{X} & \bm{x} \\
               \bm{x^{T}}& 1
      \end{array}\right ]\succeq 0,                              \end{align*}
where $T$ designates transposition of vector. The basic idea of SDR is introducing $\bm{X}=\bm{x^{T}x}$ and relaxing the equivalent constraint $\bm{X}=\bm{x^{T}x}$ to a convex one $\bm{X}\succeq \bm{xx^{T}}$. Then the SDR problem can be solved by using the interior point method with the worst case complexity $\mathcal{O}(n^{6})$. By solving the SDR problem, a lower bound of the optimal objective value of the original QCQP problem is obtained \cite{SDR}.

\subsubsection{Accelerated Algorithm}
As the product of $x_{ij}$ and $y_{jk}$ exists in both the objective and constraints, master problem $\mathcal{P}$3 is a QCQP problem. To efficiently obtain the optimal solution of problem $\mathcal{P}$3, we relax $x_{ij}(\forall i\in \mathcal{U}, j\in\mathcal{B})$ and $y_{jk}(\forall  j\in\mathcal{B}, k\in \mathcal{F})$ to continuous variables ranging between [0,1]. Such relaxation means that a file can be delivered to a user who requests it through multiple SBSs. Meanwhile, file placement relaxation means   file placement becomes a probability event.

At $t$-th iteration, problem $\mathcal{P}$3 can be converted to a relaxed problem $\mathcal{P}3'$ with SDR method.
\begin{align}\label{P3SDP}
   \mathcal{P}3':\min_{\bm{z},\bm{Z}} &~ Trace\{\bm{A_{0}\bm{Z}}\}\\
  \text{s.t.}
  &~Trace\{\bm{A_{t_1}\bm{Z}}\}+\bm{b_{t_1}}^{T}\bm{z}\preceq 0, t_1= 1,...T_1, \tag{19-a}\label{cons:RelaxFeasible}\\
  &Trace\{\bm{A_{t_2}'\bm{Z}}\}+\bm{b_{t_2}}'^{T}\bm{z}\preceq 0, t_2= 1,...,T_2, \tag{19-b}\label{cons:RelaxInFeasible}\\
  &~0\leq\bm{r}_l^{T}\bm{z}\leq 1,1\leq l\leq UB,                         \tag{19-c}\\
  &~ \bm{c}_i^{T}\bm{z}=1,\forall i\in \mathcal{U},                  \tag{19-d}\label{cons:UniqueAsso}\\
  &~0\leq\bm{s}_m^{T}\bm{z}\leq 1,1\leq m\leq BF,                         \tag{19-e}\\
  &~ \bm{d}_j^{T}\bm{z}\leq M_j,\forall j\in \mathcal{B},           \tag{19-f}\label{cons:Backhaul}\\
  &~Trace\{\bm{B}_{j}\bm{Z}\}+\bm{b}_j^{T}\bm{z}\leq C_{j},\forall j\in \mathcal{B},   \tag{19-g}\label{cons:Caching}\\
  &~\left[ \begin{array}{cc}
        \bm{Z} & \bm{z} \\
        \bm{z^{T}}& 1
        \end{array}\right ]\succeq 0,                               \tag{19-h}\label{cons:Matrix}
\end{align}
where $\bm{z}_{(UB+BF+1)}$ and $\bm{Z}_{(UB+BF+1)\times(UB+BF+1)}$ are the variable vector and matrix in problem $\mathcal{P}3'$ respectively. All the coefficients  in problem $\mathcal{P}3'$ correspond to those in problem $\mathcal{P}3$. In detail, coefficient matrix $A_{t_1}$ and vector $b_{t_1}^{T}(t_1= 1,2...T_1)$ in  (\ref{cons:RelaxFeasible}) correspond to those in  (\ref{cons:Feasible}) of problem $\mathcal{P}$3.
Coefficient matrix $A_{t_2}'$ and vector $b_{t_2}'^{T}(t_2= 1,2...,T_2)$ in  (\ref{cons:RelaxInFeasible})   correspond to those in (\ref{cons:InFeasible}) of problem $\mathcal{P}$3.
Coefficient vectors $\bm{r}_l^{T},1\leq l\leq UB$ in (19\text{-}c) and $\bm{c_i}~(\forall i \in \mathcal{U})$ in (19\text{-}d) correspond to those in   (16\text{-}c) and (16\text{-}d), respectively.
Coefficient  vectors $\bm{s}_m^{T},1\leq m\leq BF$ in (19\text{-}e)  and $\bm{d_j}~(\forall j \in \mathcal{B})$ in (19\text{-}f) correspond to those in   (16\text{-}e) and (16\text{-}f) in $\mathcal{P}$3. Coefficient  matrix $\bm{B}_{j}$  and vector $\bm{b}_{j}~(\forall j \in \mathcal{B})$ in (19\text{-}g) correspond to those in   (16\text{-}g).
Hence, according to SDR definition, problem $\mathcal{P}3'$ becomes a convex problem that can be efficiently solved by the interior point method.

\begin{algorithm}
  \caption{SDR-based Method to Find Solution of Master Problem $\mathcal{P}3$}
   \label{alg2}
   \KwIn{Constant coefficient matrixes and vectors in Relaxed Master Problem $\mathcal{P}3'$: $~\bm{A_0},\bm{b_0},\bm{A_{t_1}},\bm{A_{t_2}}',\bm{b_{t_1}},\bm{b_{t_2}}',\bm{r_l},\bm{c_i},\bm{s_m},\bm{d_j},\bm{B_{j}},$ $\bm{b}_j,M_j,C_{j},~\forall l=1,...,UB,~\forall i=1,...U;$
   $~\forall m=1,...,BF,~\forall j=1,...B.$}
   \KwOut{$~\bm{z}^{*}$=(optimal user association $\bm{x}^{*}$,file placement $\bm{y}^{*},\bm{\phi}^{*}$).}
   Solve relaxed master problem $\mathcal{P}3'$ (\ref{P3SDP}) by interior point method;\\
   Use drawing random points method to generate approximate solutions to the original  master problem $\mathcal{P}3$ (\ref{MasterProblem})\cite{SDR1,SDR2};\\
   Get optimal solution $\bm{z}^{*}=(\bm{x}^{*},\bm{y}^{*},\bm{\phi}^{*})$ where optimal user association and file placement are obtained.\\
\end{algorithm}

An SDR-based accelerated algorithm is proposed to replace the solution to master problem $\mathcal{P}$3 in Algorithm 1 (from step 11 to step 13), which is described in Algorithm 2. The complexity of Algorithm 2 is $\mathcal{O}((UB+BF)^{6})$ where $U$, $B$ and $F$ are the total number of users, SBSs and files respectively \cite{SDR}. Here, we define the PUF algorithm with the proposed accelerated algorithm as \textbf{Accelerated PUF(A-PUF)}.

\section{NUMERICAL SIMULATIONS}
In this section, simulations are performed to validate our work. Firstly, the convergency and optimality of the proposed algorithm are verified. Then, the performance of the proposed algorithm is evaluated in terms of total power consumption and file delivery delay under different DSCN scenarios, compared with existing policies. The results demonstrate the advantages of joint power control, user association and file placement in the proposed algorithm.
\renewcommand\arraystretch{0.7}
\begin{table}[h]
\centering
\caption{Simulation parameters}
\label{parameters}
 \begin{tabular}{ccc}
 \hline
 \hline
 Parameters                             &  SBS\\
 \hline
 System Bandwidth                       &20 MHz\\
 Subchannel bandwidth                   &200kHz\\
 Pathloss                               &$140+36.7log_{10}d$\\
 Shadowing Deviation                            &4 dB\\
 Noise Power Density                            &-174 dBm/Hz\\
 Number of BSs                                  &50\\
 Circuit power at each SBS                      &5.1W\\
 Maximum Transmit Power                         &30 dBm\\
 Cache size of each SBS                         &30GB\\
 Backhaul bandwidth capacity                    &1Gbps\\
 Caching power coefficient              &$6\times10^{-12}$W/bit\\
 Backhaul power coefficient             &$4\times10^{-8}$W/pbs\\
 \hline
 \hline
 \end{tabular}
\end{table}

\subsection{Simulation Setup}
In each simulation, DSCN is made up of  over $50$ SBSs, serving 250 users.  The coverage area of DSCN is a square area of 250$\times$250~$m^{2}$. Locations of SBSs and users follow  Poisson point process (PPP) model. A co-channel DSCN is considered, where  channel gain between a user and a SBS includes path loss, shadow fading and antenna gain. The backhaul delay $D_B$ is  between 0.5 and 3s\cite{D-F}. The cache capacity is set to 30GB and the backhaul capacity is set to 1Gbps. There are $1000$ files in file library $\mathcal{F}$. Each file size is set to 10$\sim$300MB and its requirement on delivery rate is set to 0.5$\sim$2Mbps. Each user has its own file preference. The user preference for files follows the kernel function \cite{UserPreference}. For each simulation result, file delivery delay is averaged over all users, i.e., $\frac{1}{U}\sum_{i\in U}d_i$. Similar to \cite{WeightedSum,WeightedSum1}, the normalized factors $\delta_{p}$ and $\delta_{d}$ are 0.002 and 0.2. Other default simulation configurations are listed in Table II, based on 3GPP specification  \cite{3GPP:Spec}.

The proposed algorithm is compared with two typical existing policies described as follows.
\begin{itemize}
  \item Content-Centric Policy(CCP) \cite{CCP}: This strategy aims to optimize transmit power with file placement and power control. It is assumed that file preference is uniform among users and each SBS caches the most popular contents until its cache is full. Power control is performed to minimize total transmit power consumption with the fixed user association.
  \item Delay-first Policy(D-F) \cite{D-F}: This strategy focuses on minimizing file delivery delay for users including wireless transmission delay and backhaul delay by jointly performing file placement and user association.
\end{itemize}

\subsection{Convergency of A-PUF}
\begin{figure}[htbp]
\setlength{\abovecaptionskip}{0.cm}

\setlength{\belowcaptionskip}{-0.cm}

\centering
\includegraphics[width=0.5\textwidth]{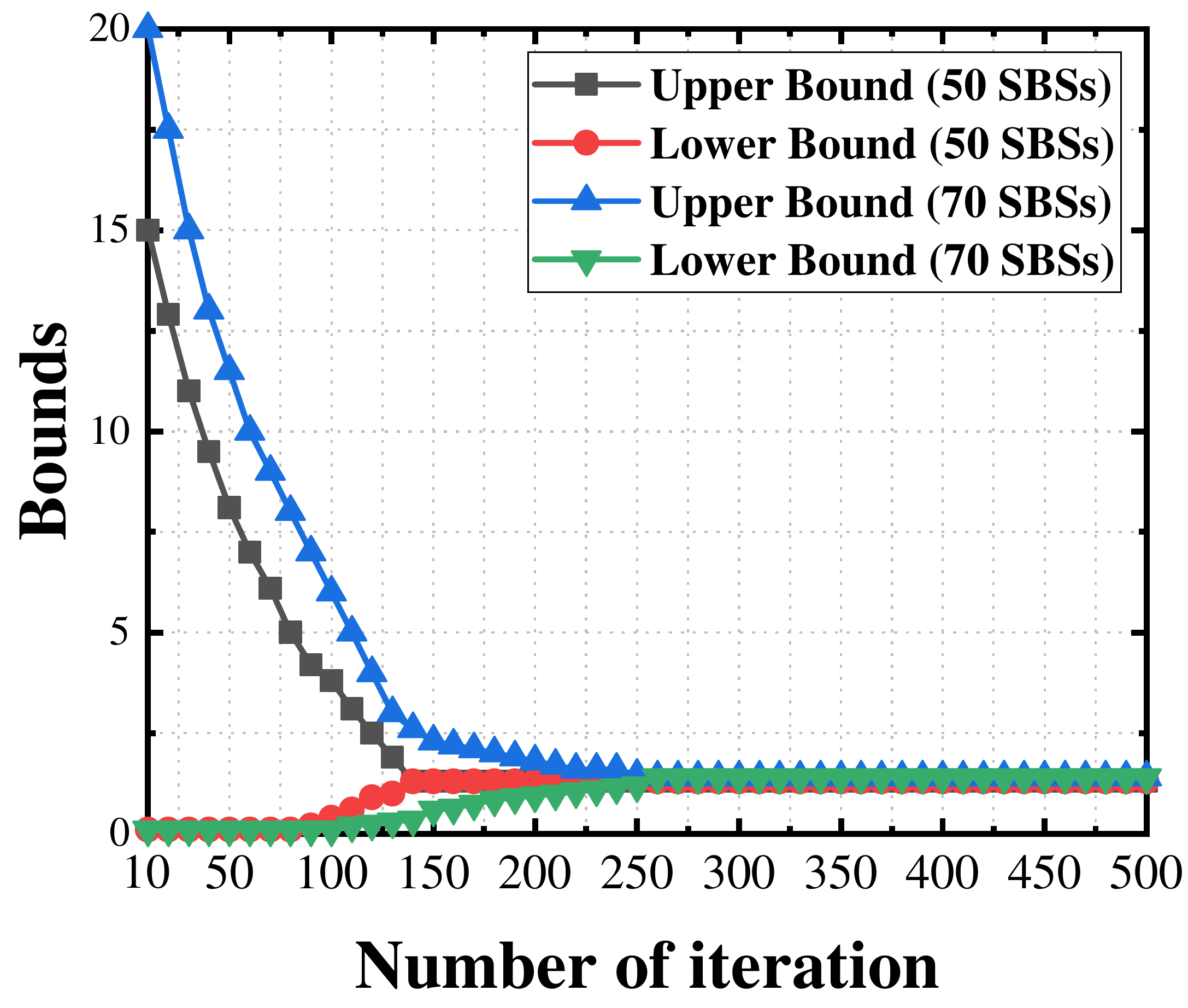}
\captionsetup{font={small}}
\caption{Upper and lower bounds when performing A-PUF}
\label{iteration}
\end{figure}
Fig. \ref{iteration} verifies the convergency of the proposed A-PUF algorithm.  The iteration number begins with 10. The cache capacity of each SBS follows a normal distribution with mean value 5000MB and the number of users is set to 250.  Balancing factor $\theta$ is set to 0.5. As expected, we can see that the upper bound and lower bound become closer with the increasing number of iterations for A-PUF. A-PUF converges to an $\epsilon$-optimal result (i.e.,$\epsilon=0.005$) after 130 iterations with 50 SBSs (800 SBSs/km$ ^{2}$) and 225 iterations with 70 SBSs (1120 BSs/km$^{2}$).

The fast convergence of the proposed A-PUF algorithm is achieved by inserting valid cuts and applying the SDR technique. As many optimal and feasible cuts as possible generated in early iterations (about 50$\thicksim$100) can largely shrink the gap between the lower and upper bounds.

\subsection{Optimality of A-PUF}
\begin{figure}[ht]
\setlength{\abovecaptionskip}{0.cm}

\setlength{\belowcaptionskip}{-0.cm}

\centering                                            
\subfigure[]{\includegraphics[width=7.46cm]{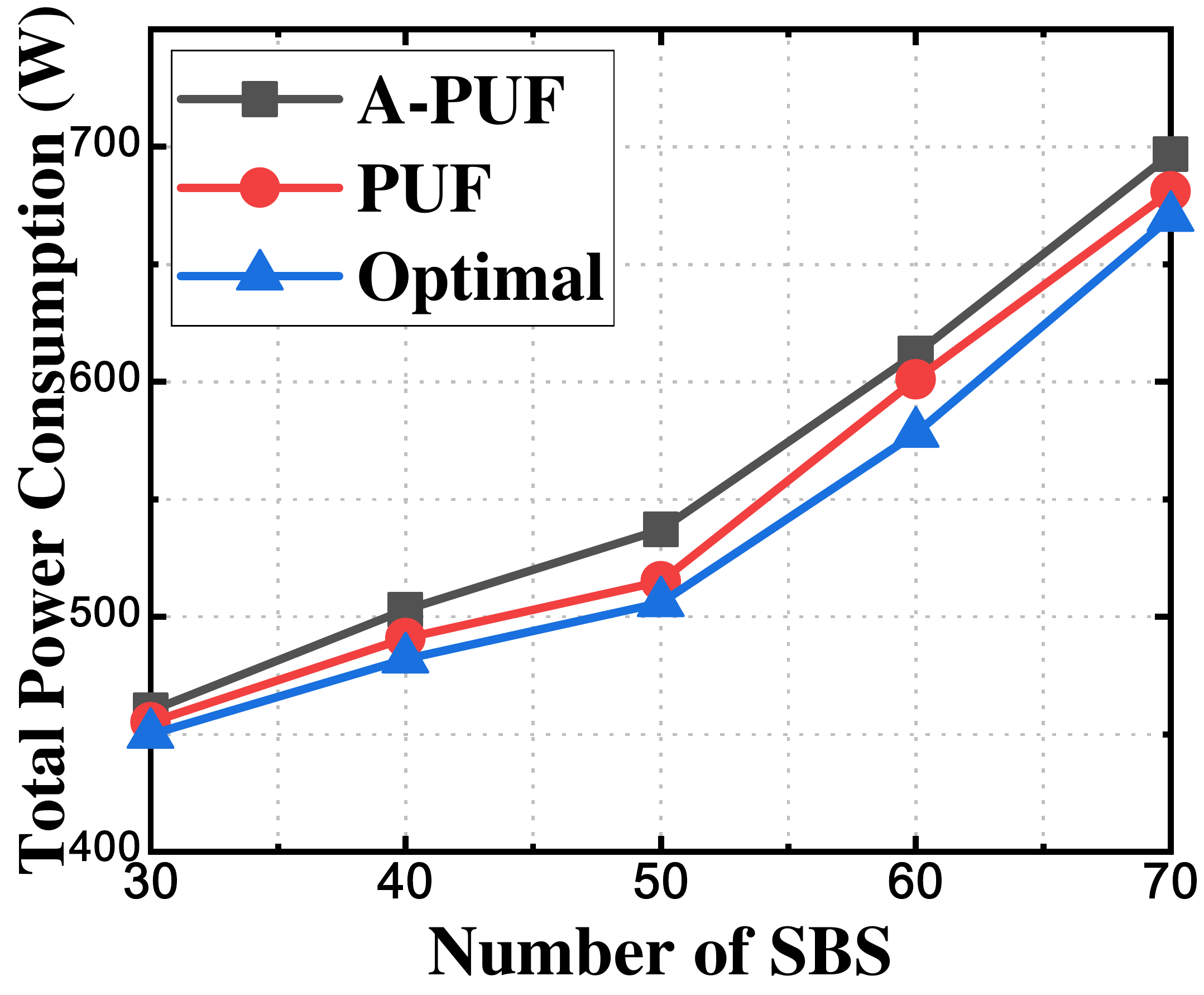}\label{EffiPower}}
\subfigure[]{\includegraphics[width=7.46cm]{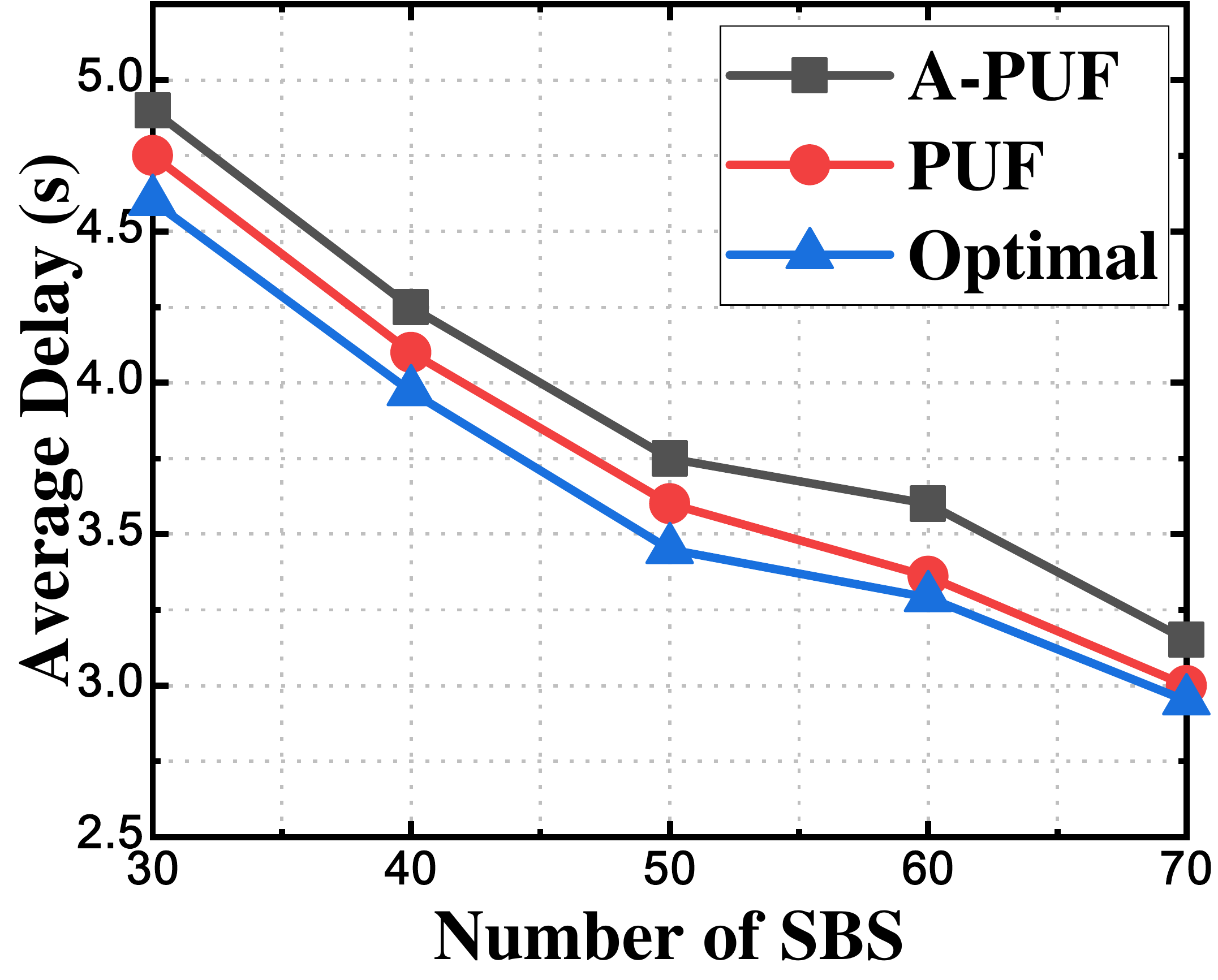}\label{EffiDelay}}                                        
\captionsetup{font={small}}
\caption{Optimality of A-PUF and PUF with different number of SBSs}
\label{Effici}
\end{figure}

To investigate the optimality of PUF and A-PUF, we use an exhaustive search algorithm to obtain optimal file delivery delay and total power consumption. Balancing factor $\theta$ is  0.5. Different DSCN sizes are considered
by varying the number of SBSs. As shown in Fig. \ref{EffiPower} and Fig. \ref{EffiDelay}, both PUF and A-PUF approach the optimal performance as the number of SBSs is varied.

For PUF, the performance loss is caused by convex approximation described in Section III. In order to obtain a convex form of $F_1{(\bm{p})}$, the downlink user data rate is relaxed based on the lower bound expressed as (\ref{appro}). This approximation results in more transmit power and larger file delivery delay compared with the optimal ones. For A-PUF, the introduction of SDR not only accelerates the convergence of the algorithm, but also incurs additional performance loss. However, compared with the significant improvement on the convergence of the algorithm, such slight performance loss is negligible.

\subsection{Performance Under Different SBS Densities}

\begin{figure}[ht]
\setlength{\abovecaptionskip}{0.cm}

\setlength{\belowcaptionskip}{-0.cm}

\centering                                            
\subfigure[]{\includegraphics[width=7.64cm]{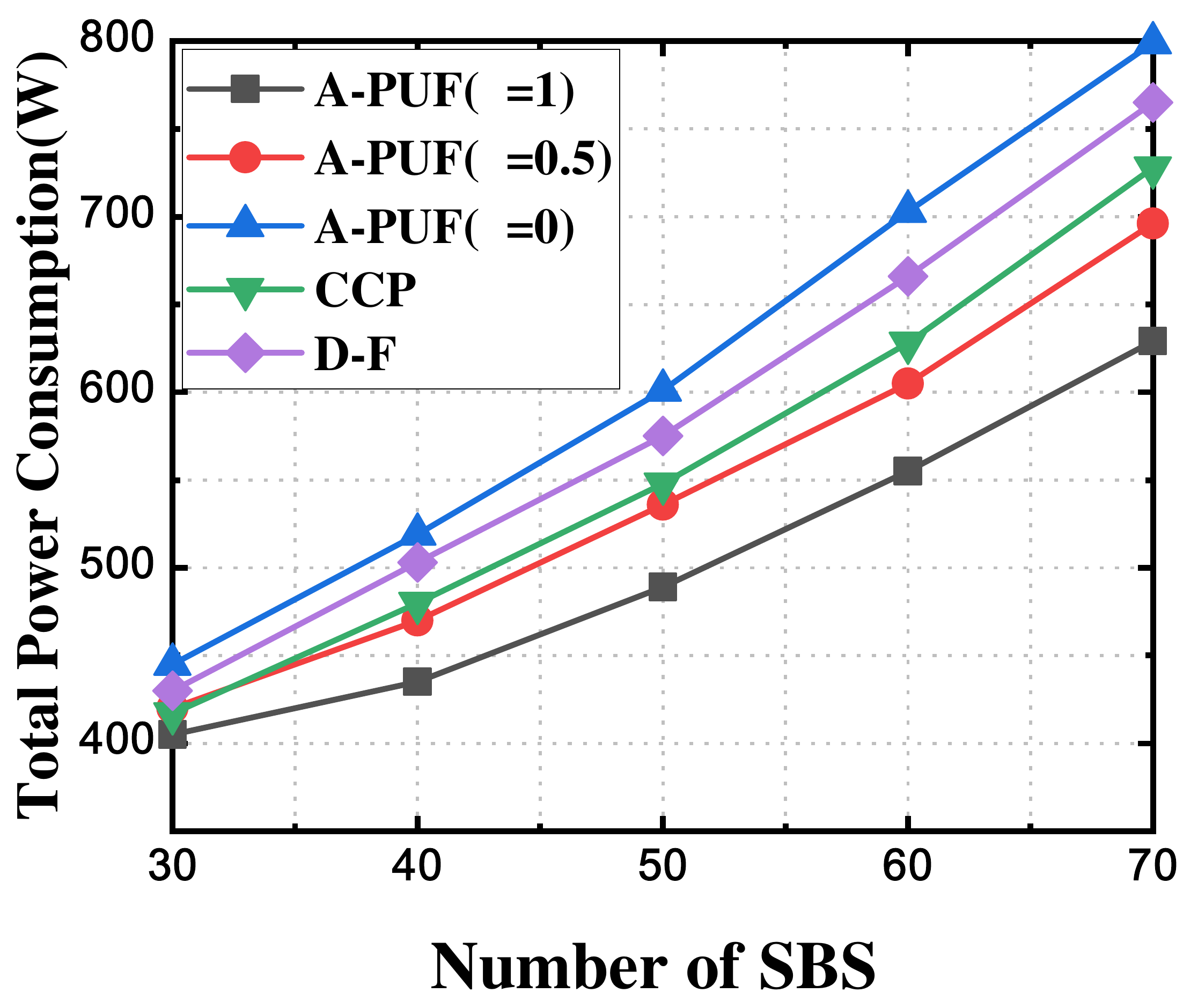}\label{PowerBS}}
\subfigure[]{\includegraphics[width=7.44cm]{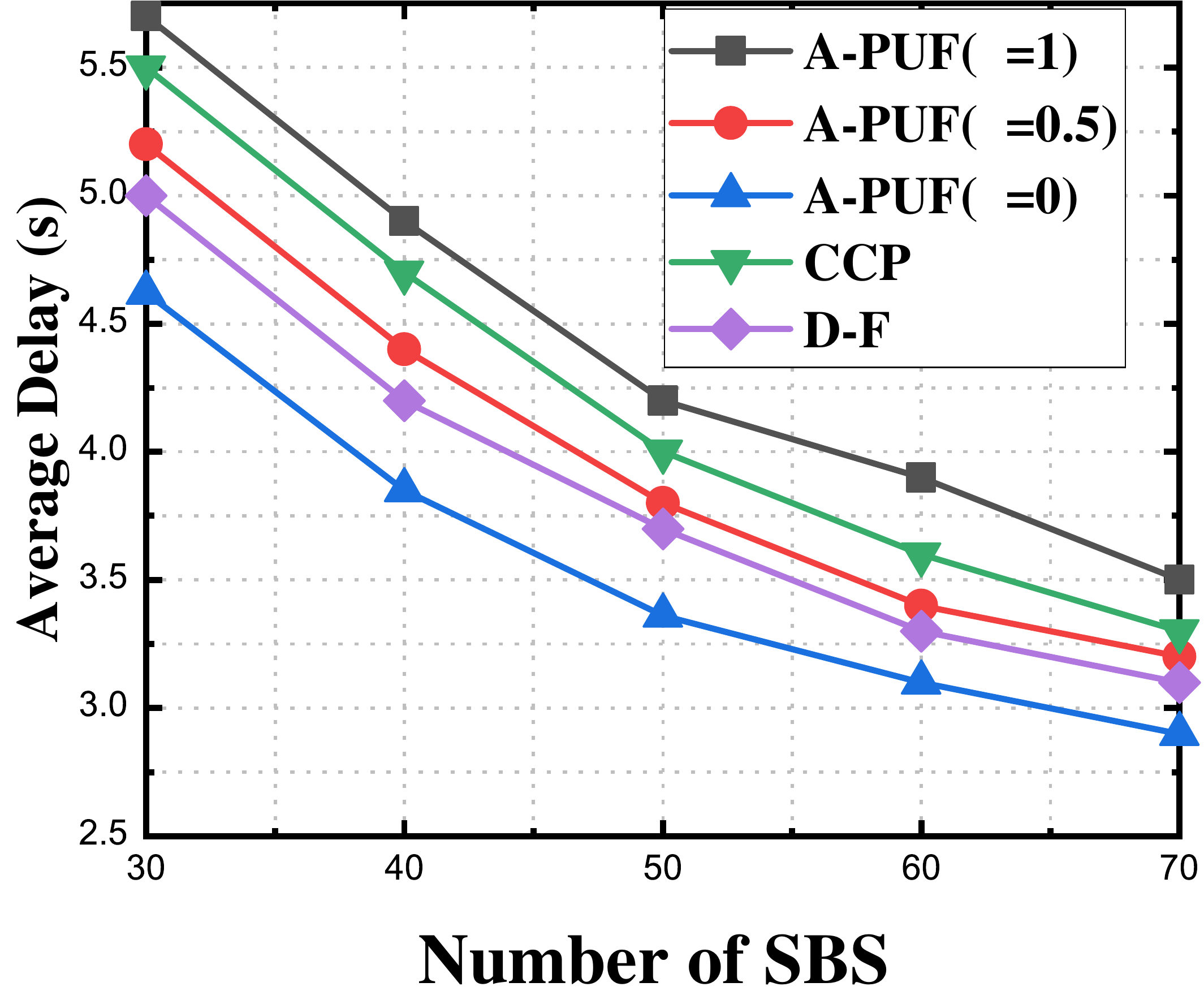}\label{DelayBS}}                                        

\subfigure[]{\includegraphics[width=7.7cm]{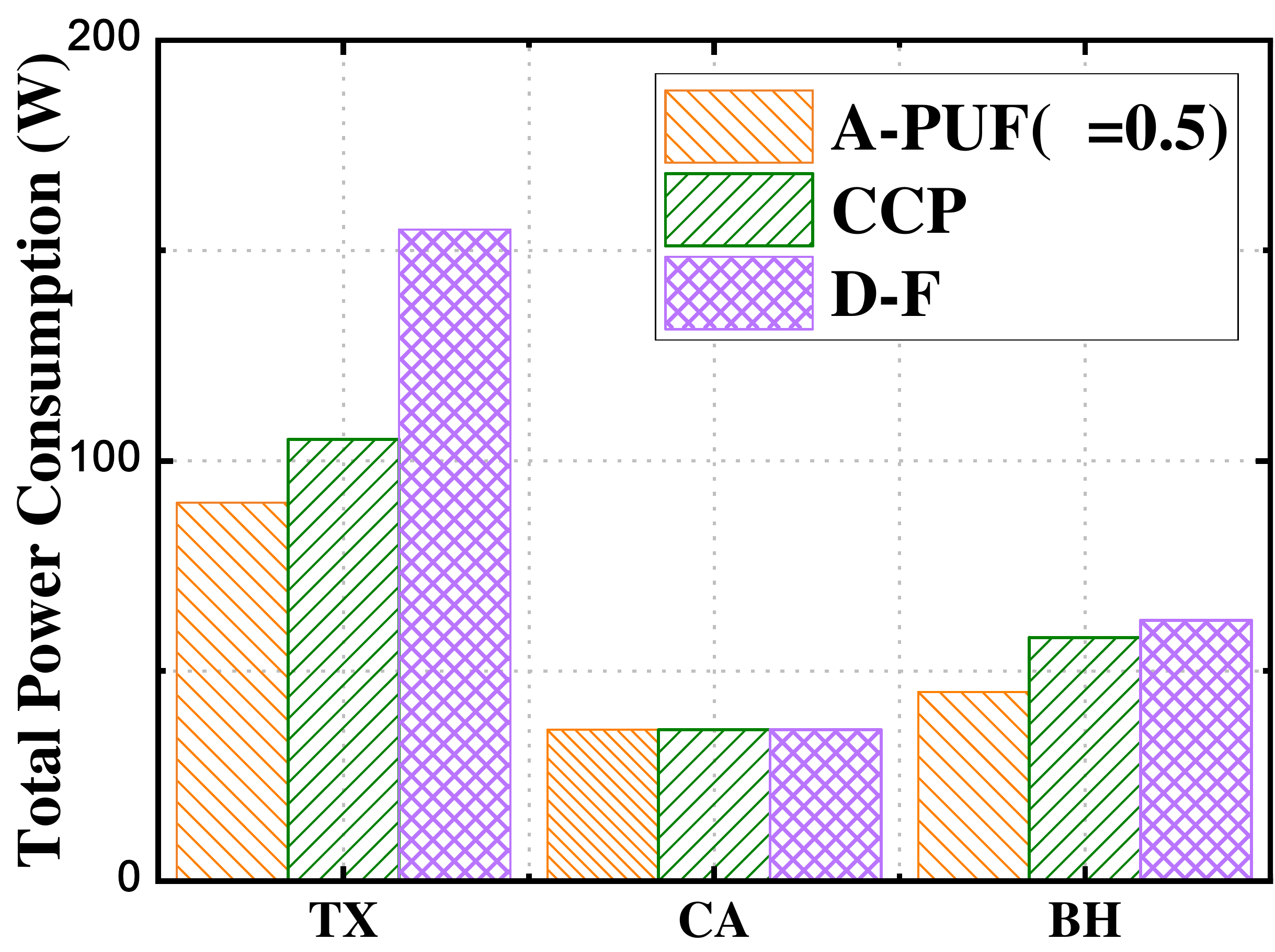}\label{Three_Power}}
\subfigure[]{\includegraphics[width=7.0cm]{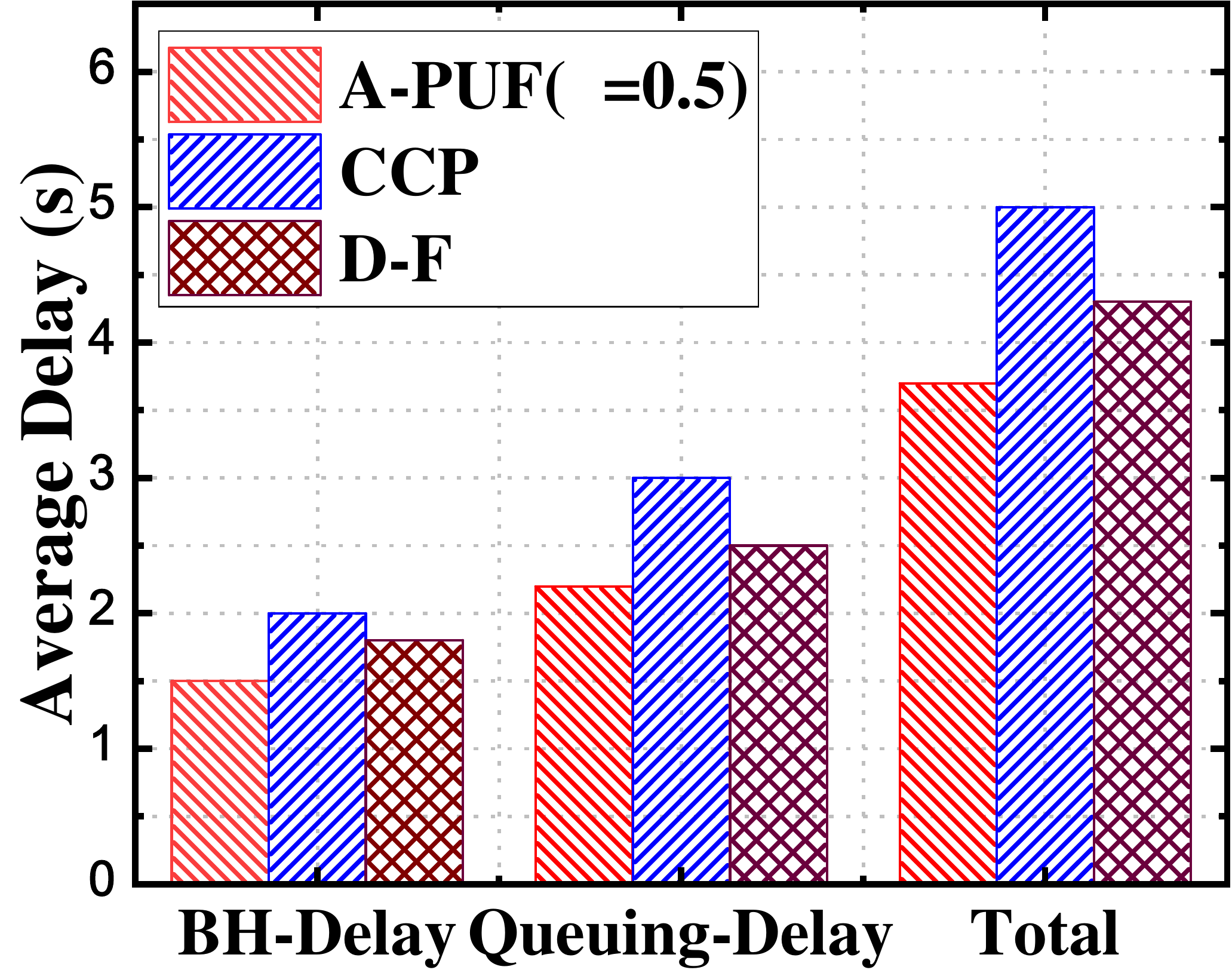}\label{Two_Delay}}                                        
\captionsetup{font={small}}
\caption{(a) Total power consumption and (b)  average delay comparison under different number of SBS
\\(c) specific power consumption and (d)  specific average delay comparison}
\label{PowerDelayBS}
\end{figure}
In the simulations, we will test our \textbf{A-PUF} algorithm with different $\theta$ values: 0, 0.5 and 1. \textbf{A-PUF($\theta$=0)} actually focuses on optimizing user delay. \textbf{A-PUF($\theta$=1)} intends to minimize the  total power consumption  called. To achieve both energy and delay minimization, $\theta$ is set to 0.5, namely \textbf{A-PUF($\theta$=0.5)}

In Fig. \ref{PowerBS} and Fig. \ref{DelayBS}, as expected, compared with other two algorithms, A-PUF($\theta$=1) and A-PUF($\theta$=0) achieve the least total power consumption and the minimum user delay, respectively. A-PUF($\theta$=0.5) can balance two objections. The proposed CCP algorithm consumes less energy than D-F but more file delivery delay.

The reason is that A-PUF($\theta$=0.5) makes full use of power control, causing more user association selection and more flexible file placement. When  transmit power is controlled among BSs, there are more  user association selection. Accordingly, the file placement at each SBS becomes more flexible. This is consistent with our research motivation in   Sec.I: power control, user association  and file placement  are coupled on power and delay optimization. In contrast, D-F results in lower file delivery delay than CCP. This is because that D-F focuses on delay minimization, which sacrifices  total power consumption. CCP results in lower total power consumption than D-F. The reason is that CCP focuses on optimizing backhaul delay and transmit power consumption with file placement and power control. However in CCP fixed user association and static file placement incurs more power consumption and user delay than A-PUF($\theta$=0.5).

Specifically, when the number of SBS is 50 and the cache capacity is 5000MB, the power consumption and average delay of different parts are shown in Fig. \ref{Three_Power} and \ref{Two_Delay}, respectively:
\begin{itemize}
  \item TX: denotes  transmit power consumption of DSCN
  \item CA: denotes caching power consumption
  \item BH: denotes  bakchaul power consumption of DSCN
  \item BH-Delay: denotes  average user backhaul delay
  \item Wireless-Delay: denotes average user wireless transmission delay
\end{itemize}
In Fig. \ref{Three_Power}, A-PUF($\theta$=0.5) consumes the least transmit power and D-F consumes the most power. This is because that power control in A-PUF can save much transmit power. An interesting observation is that all three strategies cache as many files as possible so that the caching power consumptions are the same, which indicates that caching more files can efficiently improve network performance. Besides, A-PUF($\theta$=0.5) consumes the least backhaul power. As power control is jointed with user association, A-PUF($\theta$=0.5) owns more user association secletion than CCP and  D-F, which causes more flexible file placement. Therefore, the backhaul power consumption is largely saved. Correspondingly, the average backhaul delay in Fig. \ref{Two_Delay} is least. Further, power control and user association can improve the transmission rate and the wireless transmission delay in A-PUF($\theta$=0.5) is significantly decreased.

\subsection{Performance Under Different Cache Capacities}
\begin{figure}[ht]
\centering                                            
\subfigure[]{\includegraphics[width=7.23cm]{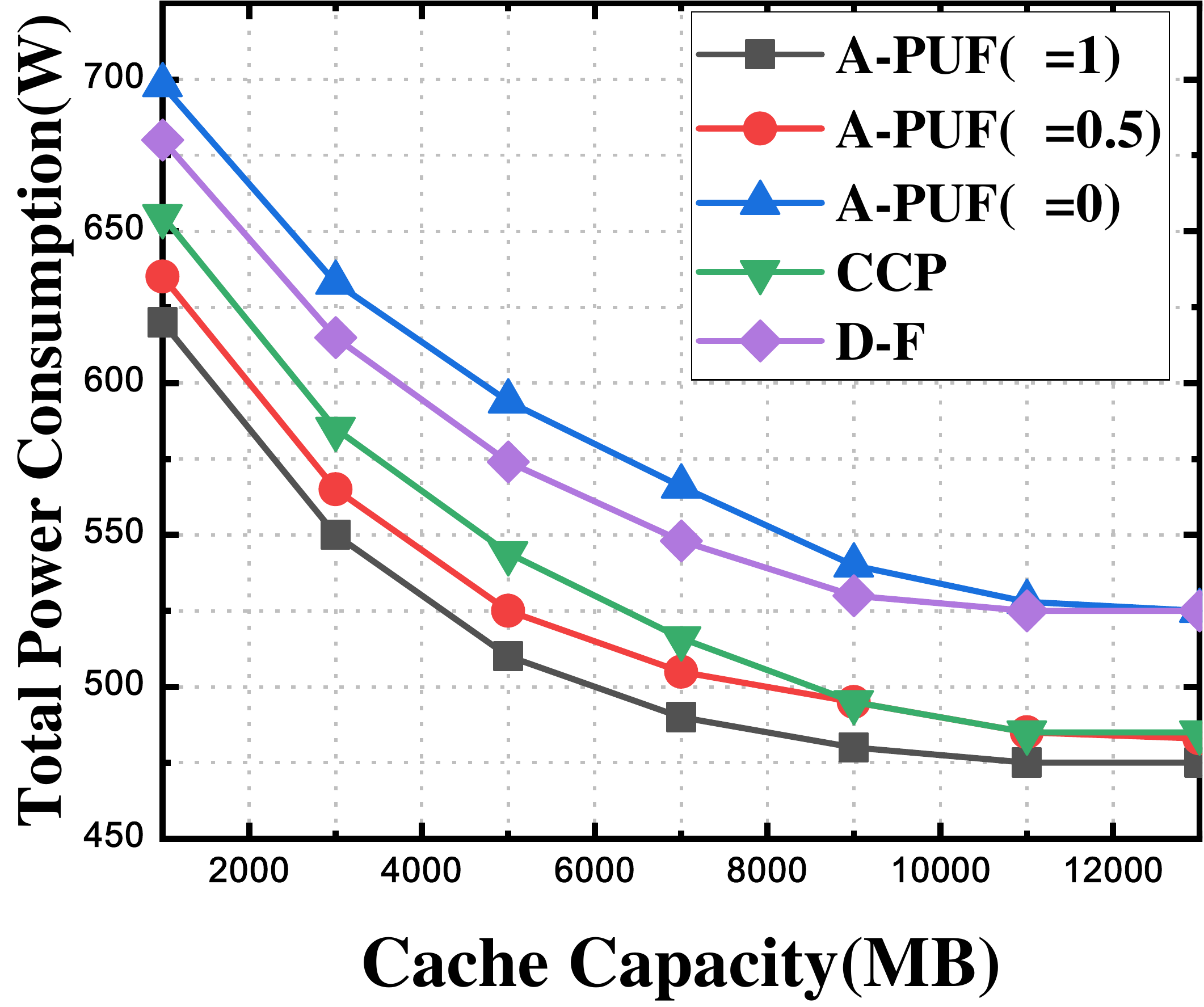}\label{PowerCache}}
\subfigure[]{\includegraphics[width=7.45cm]{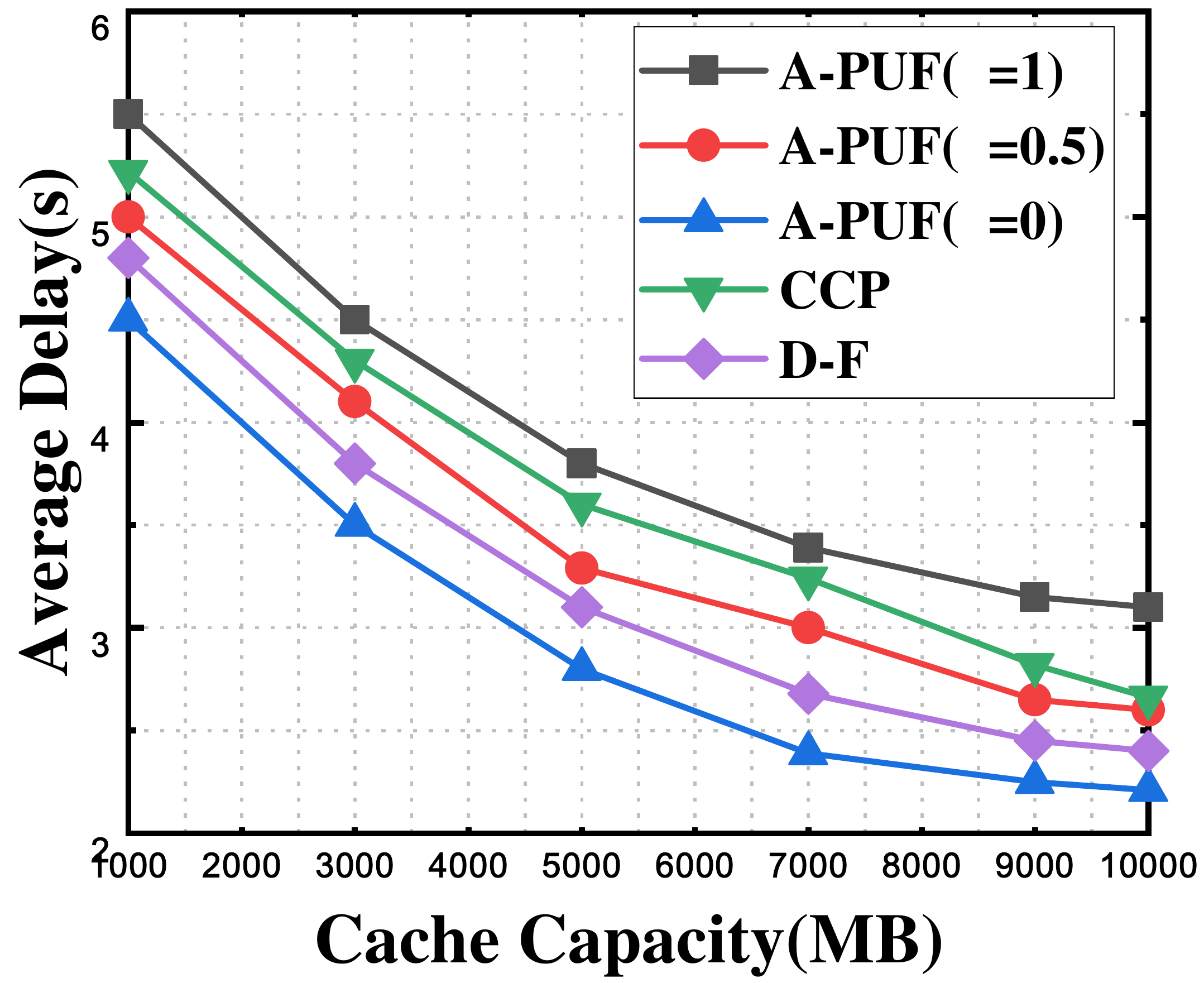}\label{DelayCache}}                                        
\captionsetup{font={small}}
\caption{(a) power consumption and (b) file delivery delay comparison under different average cache capacities}
\label{PowerDelayCache}
\end{figure}
 In Fig. \ref{PowerDelayCache}, the performance of CCF, D-F and A-PUF algorithms under different cache capacities is shown. As cache capacity increases, both  power consumption and  delay of all algorithms decrease. This is because when cache capacity increases, users can get more desired files from nearer SBSs directly and more delay(including backhaul delay and transmission delay) and  power consumption are saved.

 By observing A-PUF($\theta$=0) and D-F, we can see that A-PUF($\theta$=0) can incur less delay but more  power consumption than D-F. This reason is that, in A-PUF($\theta$=0) transmit power is controlled to improve the wireless datarate so that more power is consumed and less wireless transmission delay is obtained. After comparing A-PUF($\theta$=0.5) and CCP, we can see A-PUF($\theta$=0.5) outperforms CCP in both delay and power consumption. That indicates joint power control, user association and file placement in A-PUF($\theta$=0.5) make users obtain as many files as possible from nearer SBSs that cache requested files. As a result, power consumption (e.g.transmission and backhaul) and  backhaul delay are reduced.

\subsection{Impact of File Placement Policies}
\begin{figure}[ht]
\setlength{\abovecaptionskip}{0.cm}
\setlength{\belowcaptionskip}{-0.cm}
\centering                                         
\subfigure[]{\includegraphics[width=7.23cm]{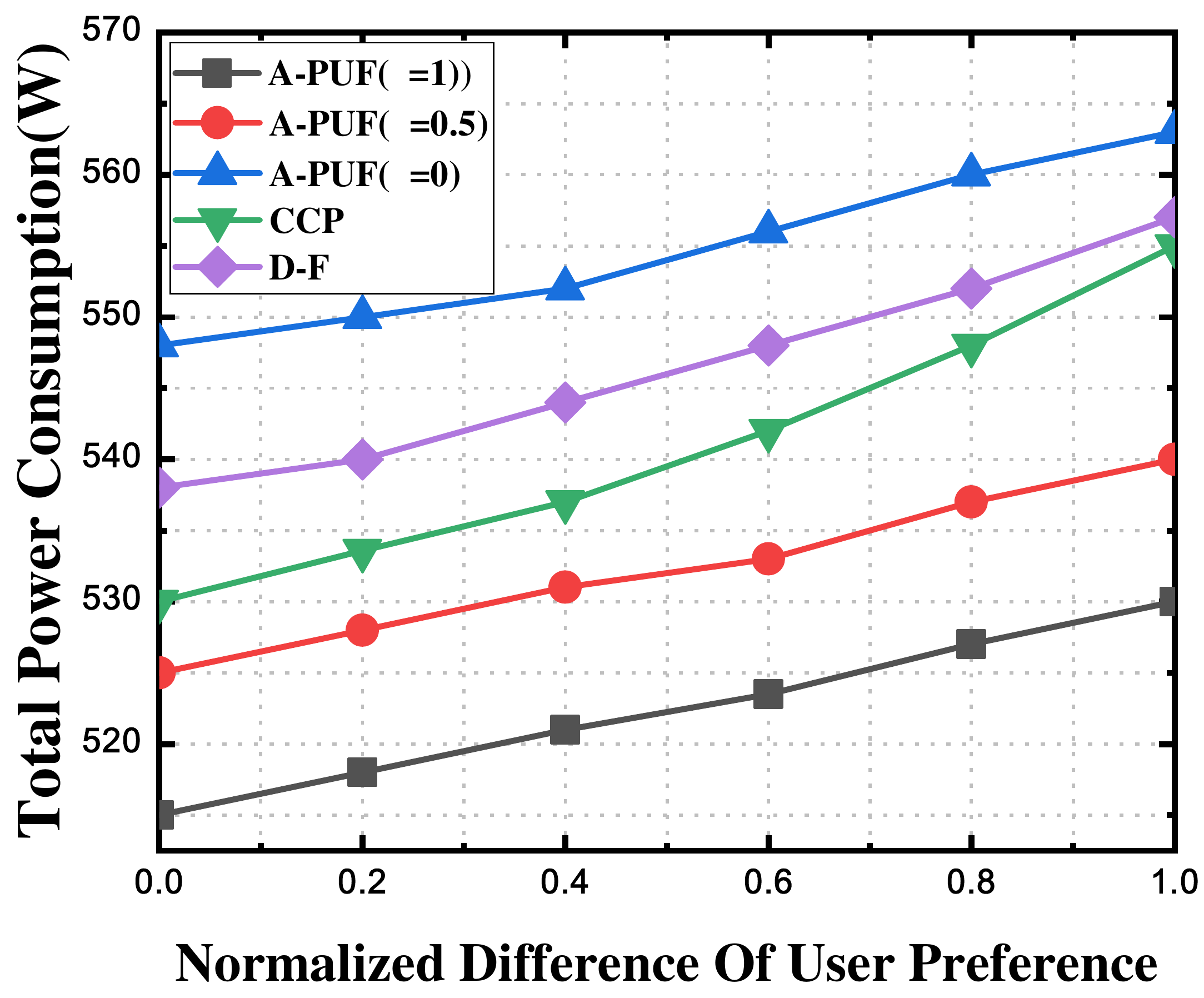}\label{PowerUserP}}
\subfigure[]{\includegraphics[width=7.4cm]{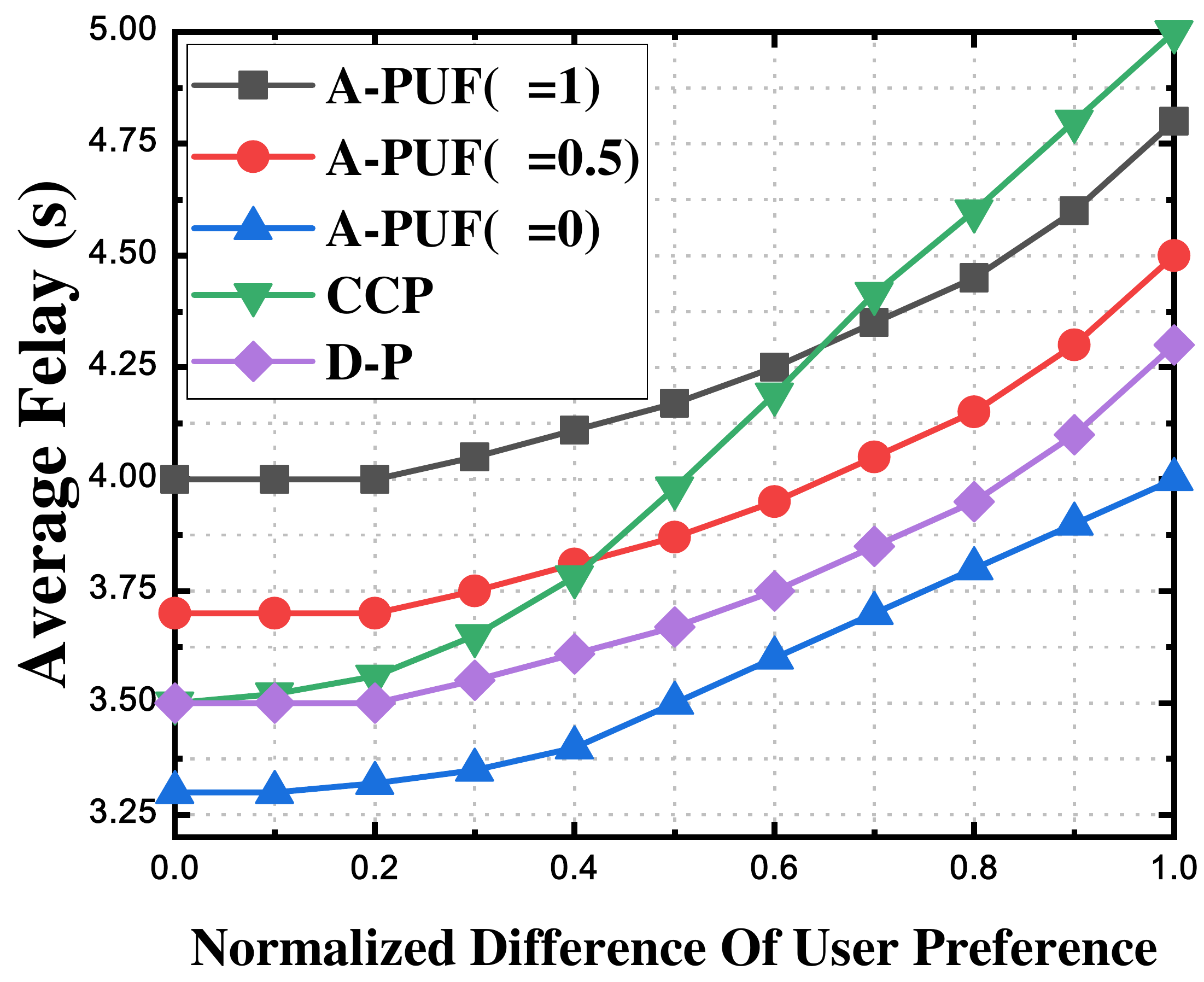}\label{DelayPre}}                                        
\captionsetup{font={small}}
\caption{Total power consumption and average delay comparison of caching strategies under different user preference}
\label{UserPreference}
\end{figure}

To verify the impact of file placement policy, we compare   performance of A-PUF, CCP and D-F under different user preference. We introduce  a parameter $Q$ to indicate difference of user preference.  First, $Q_k=\frac{\sum_{i\in\mathcal{U}}(q_{ik}-\overline{q_{k}})^{2}}{U}$ of file $f_k$  and   $Q=\sum{Q_k}$ are calculated where $\overline{q_{k}}=\frac{\sum_{i}q_{ik}}{U}$. Larger $Q$ means that the user preference are more different from each other. We normalize $Q$ and set five values from $0$ to $1$.

In Fig.\ref{UserPreference}, as normalized $Q$ increases, the power consumptions and delay of all algorithms slowly increase except CCP, which however increases rapidly. That is because, for CCP, SBSs cache the same files, ignoring  different user preference and incurring lower file hit ratio. When user preference are rather different from each other, files will be delivered by bakchaul which results longer delay and backhaul power consumption. However, thanks to the user preference-based file placement, A-PUF algorithm can maintain a steady performance(e.g. both delay and power consumption) gain in spite of the dynamic user preference.

\subsection{Delay and Power Consumption Tradeoff}

\begin{figure}[ht]

\centering                                            
\subfigure[]{\includegraphics[width=7.4cm]{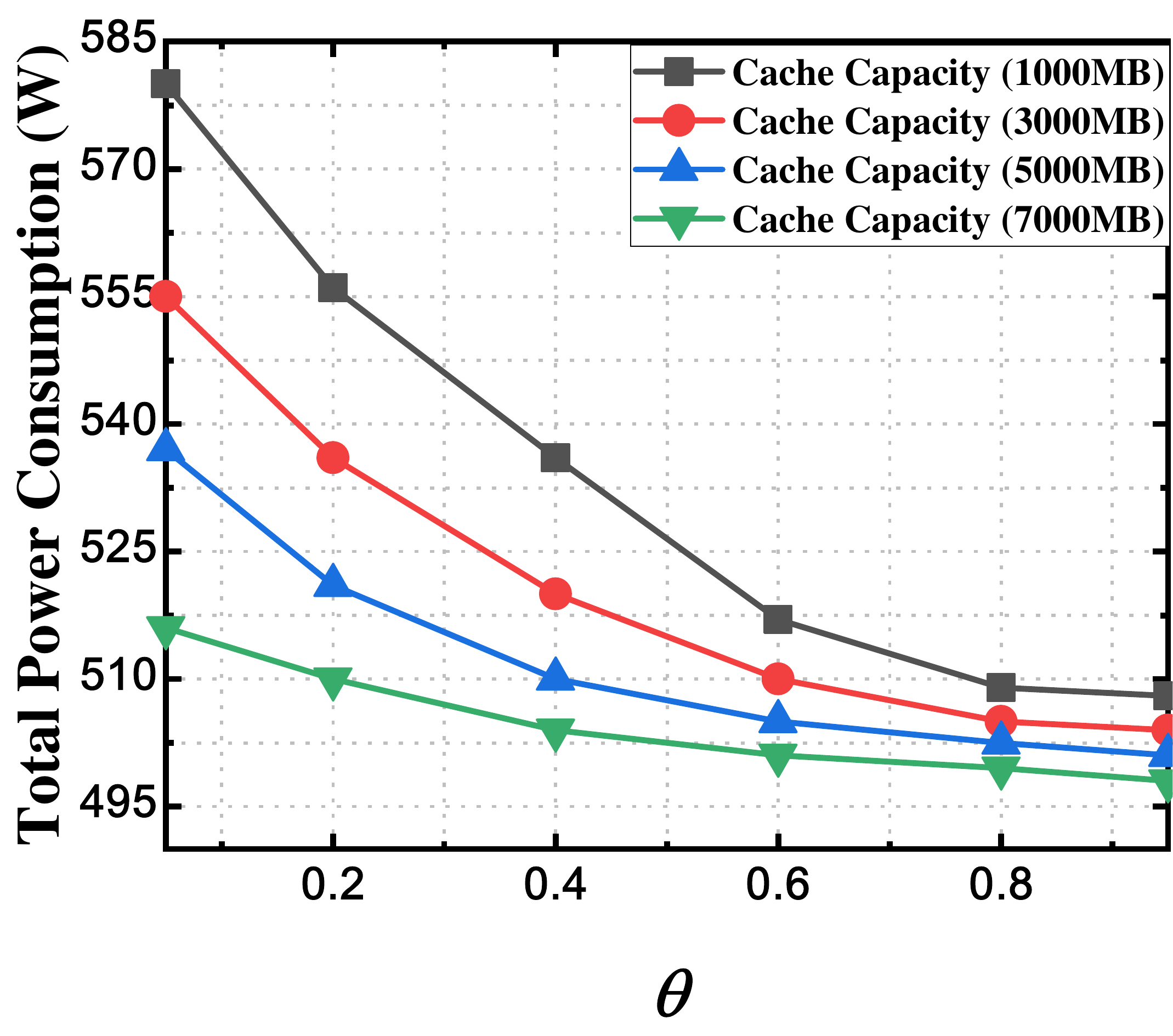}\label{TradePower}}
\subfigure[]{\includegraphics[width=7cm]{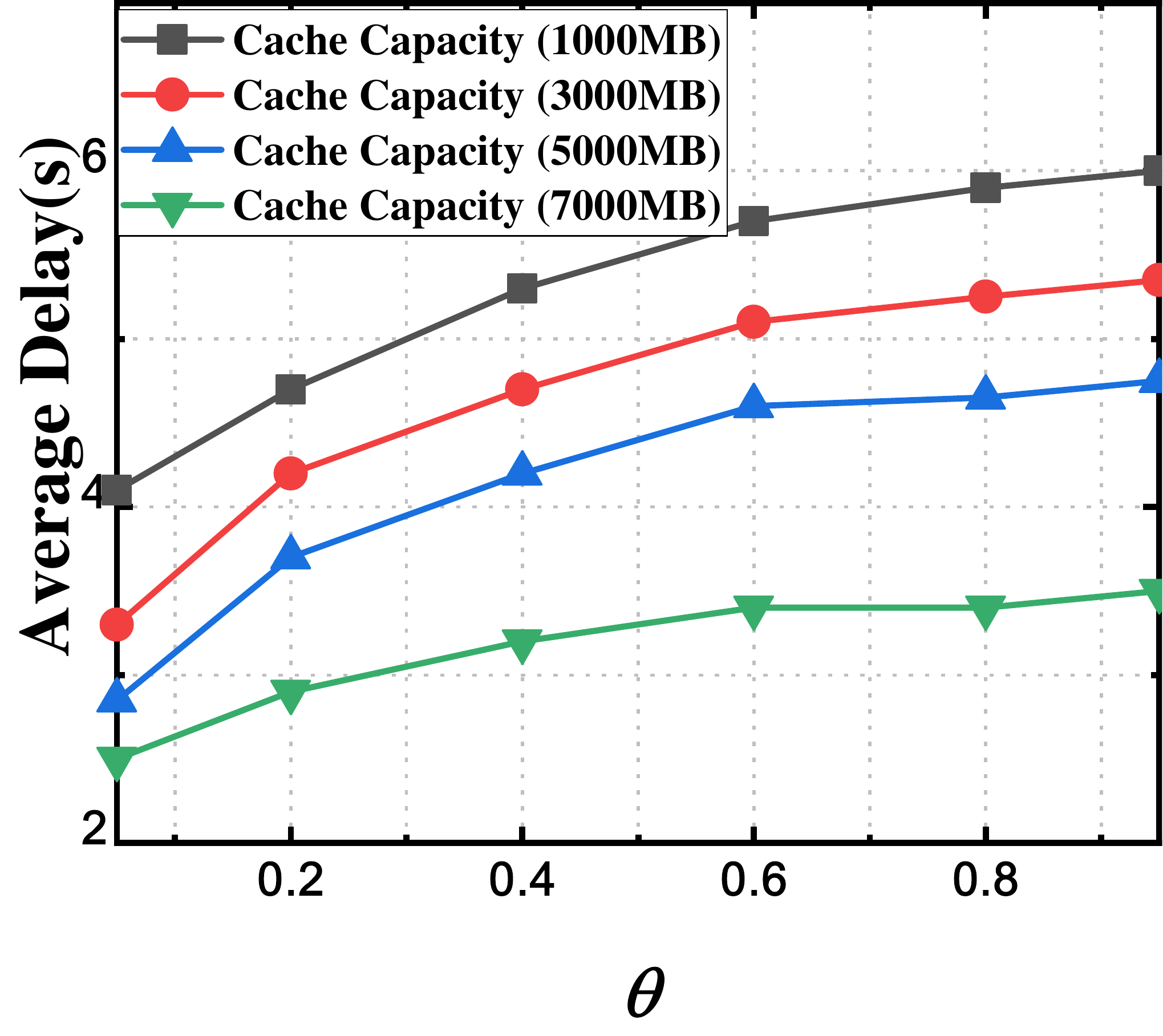}\label{TradeoffDelay} }                                        
\captionsetup{font={small}}
\caption{(a) Total power consumption and (b) average delay comparison under different balancing factor $\theta$}
\label{CachePDComparision}
\end{figure}
To verify the coupling relationship between file delivery delay and power consumption, we vary balancing factor $\theta$ under different cache capacities.  Fig. \ref{CachePDComparision} shows the power-delay tradeoff curves by adjusting  $\theta$ from $0.05$ to $0.95$ where the number of SBS is set to 50. As $\theta$ increases, power consumption is in a decreasing trend and file delivery delay in a increasing trend. Such opposite trends indicate that a desired tradeoff between file delivery delay and power consumption  can be achieved by adjusting $\theta$ to a specific value. For example, as shown in Fig. \ref{CachePDComparision}, in order to improve user experience, file delivery delay can be reduced by average $15\%$ by setting $\theta$ from 0 to 0.4. In this case, power consumption is increased by average $25\%$.

\section{conclusion}
In this paper, we solve the JDPO problem by jointly performing power control, user association and file placement. To reduce the complexity of the JDPO problem, we convert it to a form that can be handled by GBD and then decomposed the converted problem into two smaller problems, i.e., primal problem related to power control and master problem related to user association and file placement. According to the GBD approach, for the converted problem, the primal problem provides an upper bound while the master problem provides a lower bound. Based on this fact, we propose an iterative algorithm to approach the optimal solution, by solving the primal problem and the master problem iteratively. The proposed iterative algorithm is proved to be $\epsilon$-optimal. To further reduce the complexity of the master problem, an accelerated algorithm based on SDR is proposed. The simulation results show that the proposed algorithm can approach the optimal tradoff between file delivery delay and power consumption.




\begin{thebibliography}{99}\scriptsize
\bibliographystyle{ieee}
\bibliography{CASSreference}

\bibitem{Udn1}
X. Ge, S. Tu, G. Mao, C. X. Wang and T. Han, ``5G Ultra-Dense Cellular Networks," {\em IEEE Wireless
Commun.}, vol. 23, no. 1, pp.72-79, Feb. 2016.

\bibitem{DSCNReuse}
M. Ding, D. L$\acute{o}$pez-P$\acute{e}$rez, G. Mao and Z. Lin, ``Performance Impact of Idle Mode Capability on Dense Small Cell Networks,'' {\em IEEE Trans. Veh. Technol.}, vol. 66, no. 11, pp. 10446-10460, Nov. 2017.

\bibitem{CISCOVNI}
``Cisco Visual Networking Index: Forecast and Trends, 2017-2022.'' [Online] Available://www.cisco.com/c/en/us/solutions/collateral/service-provider/visual-networking-index-vni/white-paper-c11-741490.html

\bibitem{statistic}
X.~Wang, M.~Chin, and V.~C.~M.~Leung,``Cache In The Air: Exploiting  Content Caching and Delivery Techniques for 5G Systems,'' {\em IEEE Commun. Mag.}, vol. 52, no. 2, pp. 131-139, Feb. 2014.

\bibitem{statistic1}
H.~Liu, Z.~Chen, X.~Tian, X.~Wang, and M.~Tao,``On Content-Centric Wireless Delivery Networks,''{\em IEEE Wireless Commun.}, vol. 21, no. 6, pp. 118-125, Dec. 2014.

\bibitem{CCP}
M.~Tao, E.~Chen, H.~Zhou and W.~Yu, ``Content-Centric Sparse Multicast Beamforming for Cache-Enabled Cloud RAN,'' {\em IEEE Trans. Wireless Commun.}, vol. 15, no. 9, pp. 6118-6131, Sept. 2016.

\bibitem{D-F}
Y. Wang, X. Tao, X. Zhang and G. Mao, ``Joint Caching Placement and User Association for Minimizing User Download Delay,'' {\em IEEE Access}, vol. 4, pp. 8625-8633, 2016.


\bibitem{TWC17}
J. Liu, B. Bai, J. Zhang and K. B. Letaief, ``Cache Placement in Fog-RANs: From Centralized to Distributed Algorithms,'' {\em IEEE Trans. Wireless Commun.}, vol. 16, no. 11, pp. 7039-7051, Nov. 2017.

\bibitem{TMC18}
W. C. Ao and K. Psounis, ``Fast Content Delivery via Distributed Caching and Small Cell Cooperation,'' in {\em IEEE Trans. Mobile Comput.}, vol. 17, no. 5, pp. 1048-1061, 1 May 2018.

\bibitem{TOC18}
Z. Yang et al., "Cache Placement in Two-Tier HetNets With Limited Storage Capacity: Cache or Buffer?," in {\em IEEE Trans. Commun.}, vol. 66, no. 11, pp. 5415-5429, Nov. 2018.



\bibitem{CachePower1}
D. Liu and C. Yang, ``Energy Efficiency of Downlink Networks With Caching at Base Stations,'' {\em IEEE J. Sel. Areas Commun.}, vol. 34, no. 4, pp. 907-922, April 2016.



\bibitem{GC18}
Z. Gu, H. Lu, D. Zhu, Y. Lu, ``Joint Power Allocation and Caching Optimization
in Fiber-Wireless Access Networks,'' in {\em Proc. IEEE GLOBECOM}, Abu Dhabi, Dec. 2018.
\bibitem{CachePower2}
F.~Gabry, V.~Bioglio and I.~Land,``On Energy-Efficient Edge Caching in Heterogeneous Networks,''{\em IEEE J. Sel. Areas Commun.}, vol. 34, no. 12, pp. 3288-3298, Dec. 2016.




\bibitem{CachingDelay}
K.~Shanmugam, N.~Golrezaei, A.~Dimakis, A.~Molisch, and G.~Caire,``Femtocaching: Wireless Content Delivery Through Distributed Caching Helpers,'' {\em IEEE Trans. Inf. Theory}, vol. 59, no. 12, pp. 8402-8413, Dec 2013.



\bibitem{JDPO1}
V. Chamola, B. Krishnamachari and B. Sikdar, ``Green Energy and Delay Aware Downlink Power Control and User Association for Off-Grid Solar-Powered Base Stations,'' in {\em IEEE Systems Journal}, vol. 12, no. 3, pp. 2622-2633, Sept. 2018.

\bibitem{JDPO2}
J. Wu, S. Zhou and Z. Niu, ``Traffic-Aware Base Station Sleeping Control and Power Matching for Energy-Delay Tradeoffs in Green Cellular Networks,'' {\em IEEE Trans. Wireless Commun.}, vol. 12, no. 8, pp. 4196-4209, August 2013.

\bibitem{JDPO3}
L. P. Qian, Y. J. A. Zhang, Y. Wu and J. Chen, ``Joint Base Station Association and Power Control via Benders' Decomposition,'' {\em IEEE Trans. Wireless Commun.}, vol. 12, no. 4, pp. 1651-1665, April 2013.

\bibitem{R1}
H. Wu and H. Lu, ``Energy and Delay Optimization for Cache-Enabled Dense Small Cell Networks,'' arXiv preprint arXiv:1803.03780, 2018.





\bibitem{CachingDelay2}
N.~Golrezaei, K.~Shanmugam, A.~G.~Dimakis, A. F. Molisch and G. Caire, ``FemtoCaching: Wireless Video Content Delivery Through Distributed Caching Helpers,'' in {\em Proc. IEEE INFOCOM}, Orlando, FL, 2012, pp. 1107-1115.









\bibitem{GBD}
Geoffrion, Arthur M. ``Generalized Benders Decomposition'',{\em Journal of optimization theory and applications} 10.4 (1972): 237-260.

\bibitem{GBD1}
J. Krolikowski, A. Giovanidis and M. Di Renzo,``A Decomposition Framework for Optimal Edge-Cache Leasing,'' {\em IEEE J. Sel. Areas Commun.}, vol. 36, no. 6, pp. 1345-1359, June 2018.


\bibitem{GBD2}
J. Krolikowski, A. Giovanidis and M. Di Renzo, ``Optimal Cache Leasing from a Mobile Network Operator to a Content Provider,'' in {\em Proc. IEEE INFOCOM}, Honolulu, HI, 2018, pp. 2744-2752.


\bibitem{UDNFig}
M. Kamel, W. Hamouda and A. Youssef, ``Ultra-Dense Networks: A Survey,'' {\em IEEE Commun. Surveys Tuts.}, vol. 18, no. 4, pp. 2522-2545, Fourthquarter 2016.


\bibitem{PowerPara1}
S. Tombaz, P. Monti, K. Wang, A. Vastberg, M. Forzati and J. Zander, ``Impact of Backhauling Power Consumption on the Deployment of Heterogeneous Mobile Networks,'' in {\em Proc. IEEE GLOBECOM}, Houston, TX, USA, 2011, pp. 1-5.
\bibitem{PowerPara2}
 3GPP, ``Further Advancements for E-UTRA Physical Layer Aspects,'' 3rd Generation Partnership Project (3GPP), TR 36.814, Mar. 2010.



\bibitem{backhauldelay}
D. C. Chen, T. Q. S. Quek, and M. Kountouris, ``Backhauling in Heterogeneous Cellular Networks: Modeling and Tradeoffs,'' {\em IEEE Trans. Wireless Commun.}, vol. 14, no. 6, pp. 3194-3206, Jun. 2015.




\bibitem{WeightedSum}
R.~T.~Marler and J.~S.~Arora, ``Survey of Multi-Objective Optimization Methods for Engineering,'' {\em Struct. Multidisc. Optim.}, vol. 26, pp. 369-395, 2004.

\bibitem{WeightedSum1}
J. Tang, D. K. C. So, E. Alsusa and K. A. Hamdi, ``Resource Efficiency: A New Paradigm on Energy Efficiency and Spectral Efficiency Tradeoff,''{\em IEEE Trans. Wireless Commun.}, vol. 13, no. 8, pp. 4656-4669, Aug. 2014.

\bibitem{Normalization}
O. Amin, E. Bedeer, M. H. Ahmed and O. A. Dobre, ``Energy Efficiency-Spectral Efficiency Tradeoff: A Multiobjective Optimization Approach,'' {\em IEEE Trans. Veh. Technol.}, vol. 65, no. 4, pp. 1975-1981, April 2016.

\bibitem{Appro1}
W. Zhao and S. Wang, ``Low Complexity Power Allocation for Device-to-Device Communication Underlaying Cellular Networks,'' {\em Proc. IEEE ICC}, Sydney, NSW, 2014, pp. 5532-5537.

\bibitem{Appro2}
J. Papandriopoulos and J. S. Evans, ``SCALE: A Low-Complexity Distributed Protocol for Spectrum Balancing in Multiuser DSL Networks,'' {\em IEEE Trans. Inf. Theory}, vol. 55, no. 8, pp. 3711-3724, Aug. 2009.

\bibitem{Convex}
S. Boyd and L. Vandenberghe, ``Convex Optimization'', Cambridge, U.K. Cambridge Univ. Press, 2004.

\bibitem{NPHard}
M. Tawarmalani and N. V. Sahinidis, ``Global Optimization of Mixedinteger Nonlinear Programs: A Theoretical and Computational Study,'' {\em Math. Program.}, vol. 99, no. 3, pp. 563-591, Apr. 2004.


\bibitem{ConsVio}
R. Fletcher and S. Leyffer, ``Solving Mixed Integer Nonlinear Programs by Outer Approximation,'' {\em Mathematical Programming}, vol. 66, no. 1-3, pp. 327-349, 1994.


\bibitem{3GPP:Spec}
3GPP TR 36.942 V12.0.0, ``Radio Frequency (RF) System Scenarios (Release 12),'' Sep. 2010.


\bibitem{IPM}
D. H. Dirk, ``Interior Point Approach to Linear, Quadratic and Convex Programming: Algorithms and Complexity,'' Springer Science \& Business, Media, 2012.


\bibitem{SDR}
K. Huang and N. D. Sidiropoulos, ``Consensus-ADMM for General Quadratically Constrained Quadratic Programming,'' {\em IEEE Trans. Signal Process.}, vol. 64, no. 20, pp. 5297-5310, Oct.15, 2016.

\bibitem{SDR1}
M. X. Goemans and D. P. Williamson, ``Improved approximation algorithms for maximum cut and satisfiability problems using semidefinite programming,'' {\em J. ACM}, vol. 42, no. 6, pp. 1115-1145, 1995.

\bibitem{SDR2}
Z. Q. Luo, N. D. Sidiropoulos, P. Tseng, and S. Zhang, ``Approximation bounds for quadratic optimization with homogeneous quadratic constraints,`` {\em SIAM J. Optim.}, vol. 18, no. 1, pp. 1-28, 2007.



\bibitem{UserPreference}
B.~Chen and C.~Yang, ``Caching Policy Optimization for D2D Communications by Learning User Preference,'' in {\em Proc. IEEE VTC Spring}, Sydney, NSW, 2017, pp. 1-6.



\end{thebibliography}
%

\end{document}